\documentclass[conference]{IEEEtran}
\IEEEoverridecommandlockouts
\usepackage{cite}
\usepackage{amsmath,amssymb,amsfonts,amsthm}
\usepackage{algorithm}
\usepackage{algpseudocode}
\usepackage{graphicx}
\usepackage{textcomp}
\usepackage{xcolor}
\usepackage{multirow}
\usepackage{threeparttable}
\usepackage{array}
\usepackage{booktabs} 

\newcommand\cmd[1]{\Comment{\textcolor{blue}{#1}}}

\newtheorem{definition}{Definition}
\newtheorem{theorem}{Theorem}
\newtheorem{lemma}{Lemma}

\def\BibTeX{{\rm B\kern-.05em{\sc i\kern-.025em b}\kern-.08em
    T\kern-.1667em\lower.7ex\hbox{E}\kern-.125emX}}
\begin{document}

\title{Leopard: Towards High Throughput-Preserving BFT for Large-scale Systems}

\author{\IEEEauthorblockN{Kexin Hu\textsuperscript{*}, Kaiwen Guo\textsuperscript{$\dagger$}, Qiang Tang\textsuperscript{$\ddagger$}, Zhenfeng Zhang\textsuperscript{*}, Hao Cheng\textsuperscript{$\dagger$} and Zhiyang Zhao\textsuperscript{$\dagger$}}
\textit{\textsuperscript{*}Institute of Software, Chinese Academy of Sciences. Email: \{hukexin, zhenfeng\}@iscas.ac.cn} \\
\textit{\textsuperscript{$\dagger$}University of Chinese Academy of Sciences. Email: \{kaiwen2016, chenghao2020, zhiyang2018\}@iscas.ac.cn} \\
\textit{\textsuperscript{$\ddagger$}The University of Sydney. Email: qiang.tang@sydney.edu.au}
	}


\maketitle

\begin{abstract}
With the emergence of large-scale decentralized applications, a scalable and efficient Byzantine Fault Tolerant (BFT) protocol of hundreds of replicas is desirable. Although the throughput of existing leader-based BFT protocols has reached a high level of $10^5$ requests per second for a small scale of replicas, it drops significantly when the number of replicas increases. 

This paper focuses on preserving high throughput as the BFT protocol's scale is increasing. 
We identify and analyze a major bottleneck to leader-based BFT protocols due to the excessive workload of the leader at large scales.  
A new metric of \textit{scaling factor} is defined to capture whether a BFT protocol will get stuck when the scale gets larger, which can be used to measure the performance of throughput and scalability of BFT protocols. 
We propose ``Leopard'', the first leader-based BFT protocol that scales to multiple hundreds of replicas, and more importantly, preserves high throughput. 
We remove the bottleneck by introducing a technique of achieving the ideal \textit{constant} scaling factor, which takes full advantage of the idle resource and balances the workload of the leader among all replicas. 
We implemented Leopard and evaluated its performance compared to HotStuff, a state-of-the-art leader-based BFT protocol. We ran extensive experiments 
with up to 600 replicas. The results show that Leopard achieves significant throughput improvements. In particular, the throughput of Leopard remains at a high level of $10^5$ when the scale is 600. It achieves a $5\times$ throughput over HotStuff when the scale is 300, 
and the gap becomes wider as the scale further increases.
\end{abstract}



\section{Introduction}
\label{section:intro}
Byzantine fault tolerance (BFT) state machine replication protocols aim to enable a set of replicas to reach consensus on a sequence of pending requests, and endure arbitrary failures (a.k.a. \textit{Byzantine faults}) from a subset of these replicas. As a fundamental primitive in distributed computing, they have been studied for decades in many classical works \cite{PSL80reaching,Schneider90}. 
Two basic security properties are \textit{safety} that honest replicas' outputs are consistent, and \textit{liveness} that honest inputs will be output by honest replicas.

BFT protocols have received a lot attentions showing substantial progresses both in security 
and efficiency. 
In particular, the optimal resilience bound (1/3) of Byzantine fault while ensuring the security is proved in the partially synchronous network model, where a known bound of message transmission holds after some unknown global stabilization time (GST) \cite{DLS88}. 
Traditional BFT solutions that guarantee the security with the optimal resilience normally support for a small scale of replicas (e.g., a dozen replicas) with moderate throughput (e.g., a few thousand requests per second) \cite{PBFT99}. 

Since the emergence of decentralized applications on the Internet started from Bitcoin \cite{Nak08bitcoin}, even consortium blockchain may have a large number of peers (e.g., in a global supply-chain); moreover, BFT protocols are often adopted as an important component for the permissionless consensus (e.g., \cite{OmniLedger18}) and  proof-of-stake protocols (e.g., Algorand \cite{Algorand17}), in which a large number of committee members (thousands) are chosen to run a BFT protocol. It follows that BFT protocols that can support a large number of replicas have become highly demanded. 
A practical BFT protocol should thus perform well on efficiency and scalability simultaneously, i.e., preserving high throughput when the number of replicas increases\footnote{Algorand \cite{Algorand17} provides a ``scalable'' BFT solution that supports thousands of nodes. However, the throughput is only about $1000$, while we would like to have both scalability and high throughput.}. 

Great efforts have been made to improve scalability and efficiency. 
While some of them achieving a better efficiency at the cost of weakening the resilience \cite{FastB05}, there are major improvements recently have realized a \textit{linear} communication complexity without sacrificing the resilience, by combining several optimization techniques both in theory and implementation \cite{SBFT19,HotStuff-imple19,BFTSmart}. Evaluations show that high throughput of confirming $10^{5}$ requests per second is achieved \cite{HotStuff-imple19}.


\begin{figure}[!h]
	\vspace{-3mm}
	\centering{\includegraphics[width=0.35\textwidth,trim=0 0 0 0,clip]{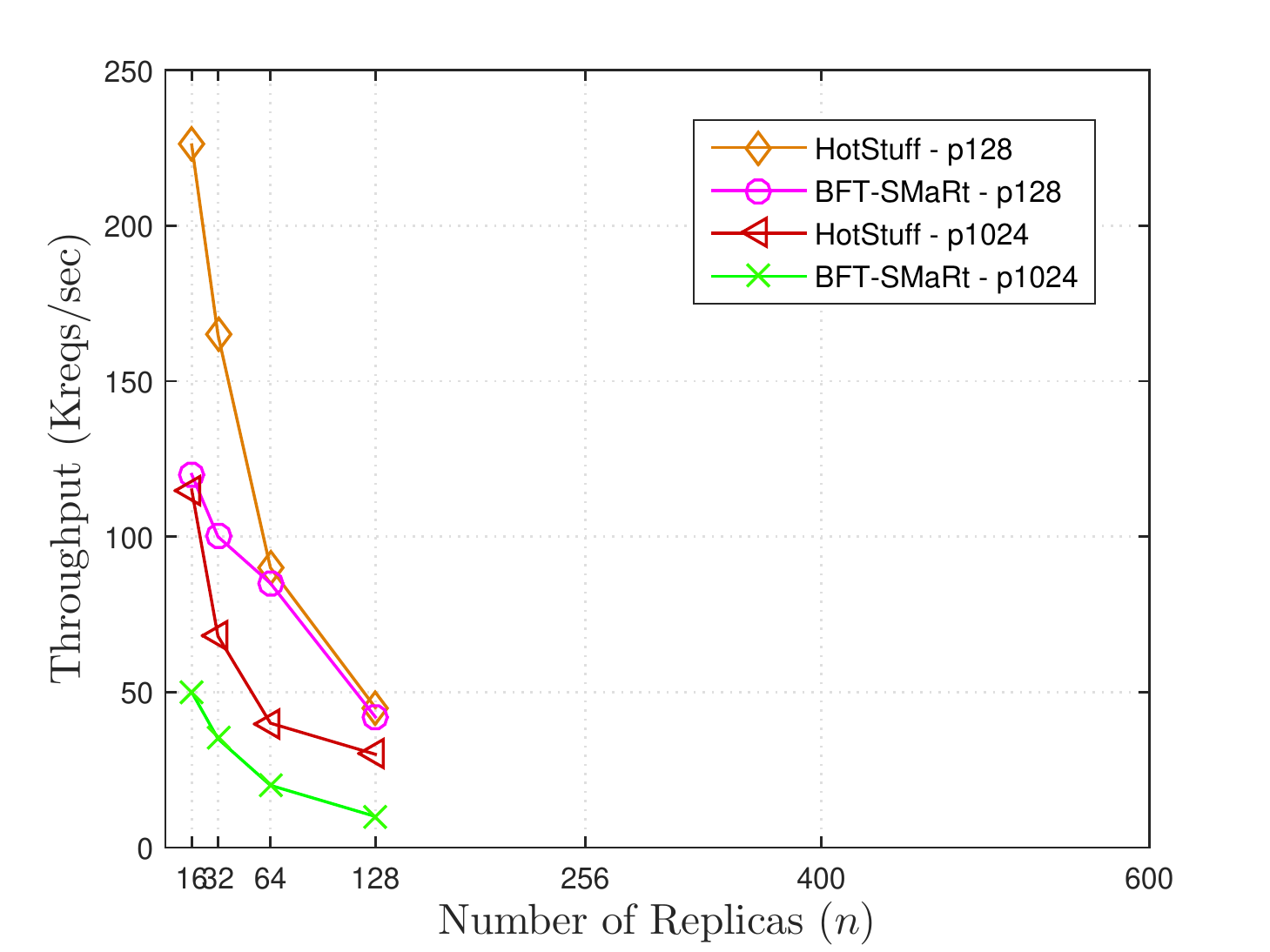}}
	\vspace{-2mm}
	\centering \caption{Evaluation result \cite{HotStuff-imple19} on scalability and throughput in HotStuff \cite{HotStuff-imple19} and BFT-SMaRt \cite{BFTSmart}, with 128-byte (p-128) and 1024-byte (p-1024) payloads.} 
	\label{fig:tpsothers}
	\vspace{-2mm}
\end{figure}

Unfortunately, the latest evaluation results (Fig.\ref{fig:tpsothers}) of the state-of-the-art protocols, HotStuff \cite{HotStuff-imple19} and BFT-SMaRt \cite{BFTSmart} (an optimized implementation based on PBFT \cite{PBFT99}) 
show that, high throughput appears only on a small scale of replicas; when the scale increases, the throughput drops sharply, even in the optimistic case that replicas follow the protocol and the network is synchronous! 


There are two common methods for improving scalability and efficiency in a distributed system, namely scaling out and scaling up. 
One popular approach to design a ``scale-out" BFT system is via sharding \cite{OmniLedger18}. By partitioning replicas into multiple shards that execute a BFT protocol to handle requests in parallel, the processing capacity scales horizontally with the number of replicas. 
To ensure a small probability that any shard has more than 1/3 fraction of Byzantine replicas, it still requires the existence of BFT protocols supporting a large scale (multiple hundreds or even thousands) of replicas for each shard, meanwhile preserving high throughput \cite{OmniLedger18}.
On the other hand, scaling up is to improve efficiency by adding more resources to each replica. Since resources are usually costly \cite{BandwidthMicrosoft}, scaling up a system is practical only if the performance improvement is significant. However, we find (see below in $\S$\ref{exp:costeffectiveness}) that the  throughput increase via scaling-up in state-of-the-art protocols approaches 0 when the scale (number of replicas) gets larger. 

These observations motivate us to consider the following main question of the paper: 
\begin{center}
{\em Can we design a partially synchronous BFT protocol that is scalable (to at least several hundred or more) under optimal resilience; more importantly, it preserves high throughput thus enabling an effective scaling up?}
\end{center}





\subsection{Our Contributions}
Most of the efficient BFT implementations, such as HotStuff and BFT-SMaRt, are leader-based constructions. 
In this paper, we also focus on scalable and efficient leader-based (partially synchronous) BFT protocols. Our contributions are three-fold:

\smallskip
\noindent\textbf{Identifying the major bottleneck.} 
Most leader-based BFT protocols follow the design paradigm of the seminal work of PBFT: the leader first \textit{proposes} a consensus proposal (e.g., a new block), containing a bunch of pending requests to other replicas, for initiating a new agreement instance; each replica then \textit{votes} on this proposal by multicasting several rounds of votes. 
Since the well-known drawback of PBFT is the heavy communication cost which is mostly due to the all-to-all communication pattern during voting, several works \cite{SBFT19,HotStuff-imple19} improves this by using cryptographic techniques to aggregate votes of size $O(n)$ into one proof of size $O(1)$ in every round of voting. Combining with the pipelining technique \cite{Casper17}, where the second round of voting on a consensus proposal is piggybacked on the next proposal's first round of voting, the amortized communication cost can be further reduced.

While those protocols mainly focus on reducing the cost of the protocol overall, one crucial but often neglected aspect is that the \textit{leader} might be overloaded. 
To validate this simple observation experimentally, we measured the leader's workload\footnote{We use the leader's bandwidth utilization as an example to show its workload since the bandwidth utilization is a critical metric and it can be easily measured.} 
with an increasing number of replicas. The result in Fig. \ref{Fig:HotStuff} shows that the leader's workload (expressed by the bandwidth utilization) increases significantly as the protocol's scale gets larger, thus at the same time, the throughput decreases dramatically. 

\begin{figure}[ht]
	\vspace{-2mm}
	\centering		
	\includegraphics[width=0.35\textwidth,trim=0 0 0 0,clip]{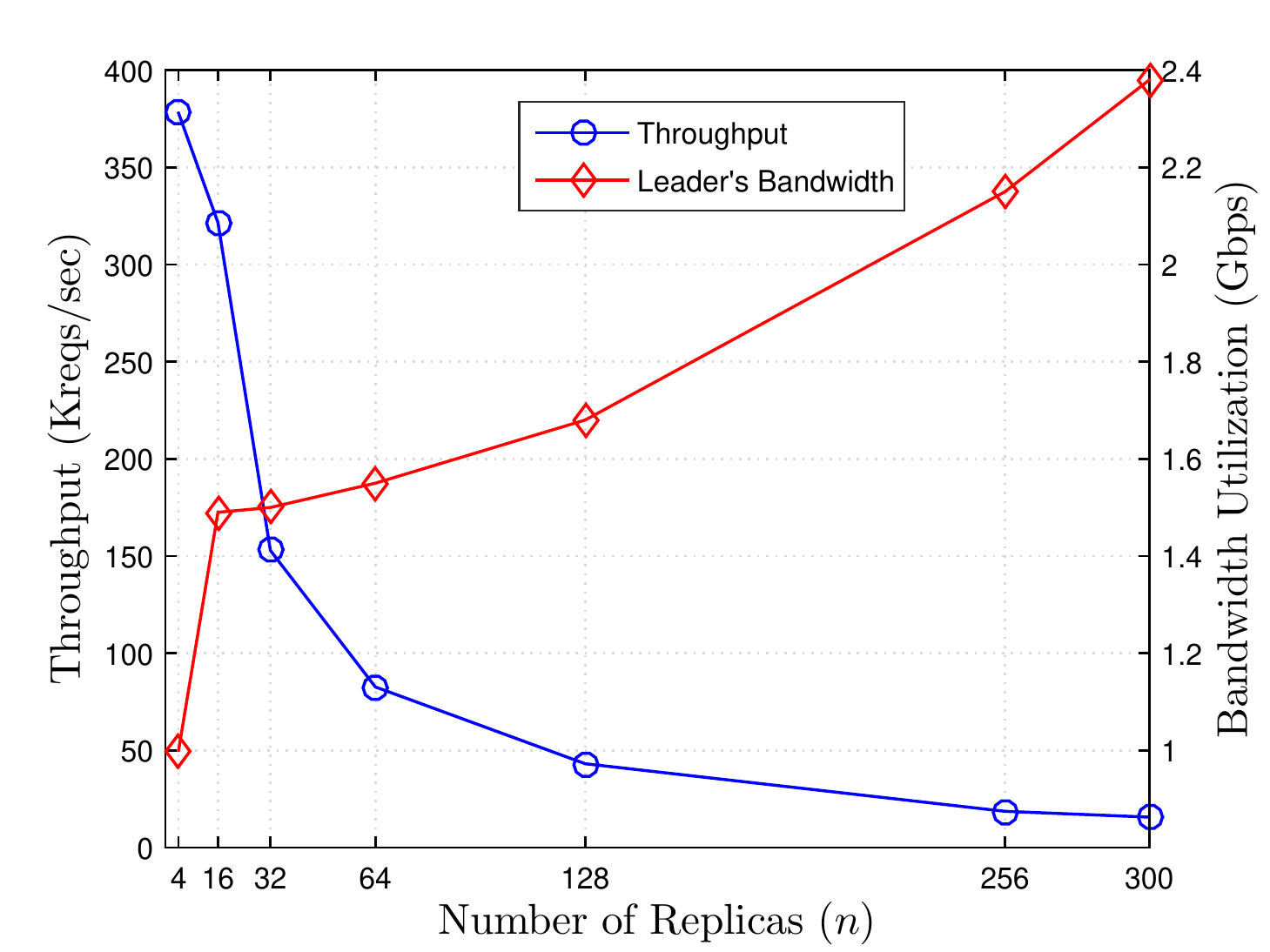}	
	\vspace{-2mm}
	\caption{Throughput and leader's bandwidth utilization in HotStuff with 128-byte payload.} 
	\vspace{-2mm}
	\label{Fig:HotStuff}
\end{figure}

Recall that the leader serves as disseminating every pending request via consensus proposals to all the other $n-1$ replicas, and each non-leader replica only receives requests from the leader. Let $\Lambda$ be throughput (the number of requests to be confirmed per second) that we want to achieve, and payload size be the number of bits per request. 
When processing $\Lambda$ requests, the leader's communication cost \textit{at the propose step} during request dissemination\footnote{A conventional optimization with engineering was running the consensus on the digest of requests to reduce the cost. But the complete requests are still required to be delivered to all replicas (otherwise, it is the equivalent of confirming requests with a lower payload size). The increase of leader's workload with an increasing $n$, however, still remains even using the optimization, and it cannot help to achieve a constant scaling factor.} is
\begin{equation}
			\Lambda\times\textrm{payload size}\times(n-1),
	\label{eq:leader-workload}
\end{equation}
(whereas a non-leader replica's communication is $\Lambda\times\textrm{payload size}$). 
Since the maximal achievable throughput is subject to the constraint of the number of bits that can be transmitted per second at each replica, denoted as $C$, throughput shall certainly drop when $n$ grows. 


To preserve high throughput when the protocol's scale increases, no replica's communication cost should become a bottleneck. To capture this, we define a performance metric called \textit{scaling factor} to be the heaviest communication cost among all replicas for {\em confirming one bit} of request following the protocol, i.e. 
$$SF=\max_{1\leq i\leq n}\{c_i\}$$ where $c_i$ denotes the communication cost of replica $i$.

We can see that a protocol's expected throughput is limited by $C/SF$.
%
Now let us examine the leader-based BFT protocols having leader cost as \eqref{eq:leader-workload} in the above. Let $c_L$ (resp. $c_R$) denote the communication cost for confirming a bit of request at the leader (resp. a non-leader replica). Clearly now {$SF^\textrm{}=\max\{c_L,c_R\}=O(n)$,} 
it follows that the expected  throughput decreases as $n$ grows (when each replica's capacity is fixed).

To preserve throughput when the protocol's scale increases, it would be desirable to design a protocol that has scaling factor increases slower (e.g. sub-linear to $n$), and ideally, just a small constant that is  irrelevant to $n$!

In addition, let us also examine the effectiveness of the scaling up solution by adding more resources, e.g., bandwidth. 
Since the number of bits that can be confirmed per second is bounded by $C$, we assume the fraction of $C$ that contributes to throughput is $\gamma$, 
where the exact value of $\gamma$ depends on the concrete protocol design but $\gamma\leq1$ always holds. 
When scaling up a system, i.e., increasing $C$, and to make the best use of the added resources, $\gamma$ should be a large value, and in the ideal case, approaching 1. 
However, it can be deduced from \eqref{eq:leader-workload} that, the $\gamma$ is at most $\frac{1}{n-1}$ in these leader-based BFT protocols. 
And it quickly approaches 0 as $n$ increases. This means that scaling up becomes ineffective for these BFT systems to improve throughput in a large-scale environment, and it leaves considerable room for improvement. 



\smallskip
\noindent
\textbf{A high throughput-preserving BFT protocol.} We present a new leader-based BFT design paradigm, and we instantiate a BFT protocol called ``Leopard'' in the partially synchronous network model. 
Our protocol addresses the main question of the paper.  
It preserves high throughput when the protocol's scale increases, while ensuring security without relaxing the optimal resilience. Notably, it achieves the desirable constant scaling factor. 
As a byproduct, Leopard also enables an effective scaling up with a constant $\frac{1}{2}\times$ improvement on throughput for all scales, when doubling the communication capacity at each replica. 
We further use the standard technique of threshold signatures for aggregating votes to reduce the communication cost. 
{As many other BFT protocols, Leopard is designed based on PBFT, by virtue of its optimal good-case round number for each decision \cite{good-case-latency21}}, numerous analyses on its security and implementations in different scenarios.  
{However, our techniques could also be applied to other advanced alternatives like HotStuff.} Let us briefly walk through the ideas. 

\smallskip
\noindent
\textit{\underline{Balancing the workload.}} 
The excessively unbalanced workloads between the leader and a non-leader replica in most leader-based protocols make the leader be the ``Achilles heel'' when the protocol scale is increasing.   
We distribute the leader's workload by taking full advantage of the idle resource of other replicas and balancing the workload among all replicas. 
To this end, we decouple a BFT consensus proposal into two planes: \textit{pending requests} and \textit{indices; each of the index points a package of pending requests},  
Every non-leader replica packs pending requests in its buffer and disseminates the packages, and the leader only has to contain links to packages in consensus proposals.

By spreading the $O(n)$ overhead for request dissemination to $O(n)$ replicas, every replica's amortized workload can remain constant (even if the protocol's scale becomes larger). 
Moreover, the decoupling of BFT consensus proposals essentially divides, from the functionality point of view, the propose step in a leader-based BFT protocol into: \textit{pending request dissemination} and \textit{new agreement instance initiation}. This enables a separation of the heavy part of request dissemination from the mainstream of consensus. Therefore, the consensus process and data delivery process can be thus executed in parallel, rather than in a sequential way as in prior. This further benefits to taking full advantage of the idle resources of replicas, and helps to improve the efficiency of the protocol.

Let $\alpha$ be the number of bits per package, $\beta$ be the size of an index, and $\kappa$ be the size per vote,   
and the scaling factor of Leopard is,  
\begin{equation*}
		SF^\textrm{Leopard}= \max\{\frac{(\beta+4\kappa)(n-1)}{\alpha}+1,\frac{(\beta+4\kappa)}{\alpha}+2\}.
\end{equation*}
Note that, $\alpha$ is a parameter of the protocol satisfying $\alpha\geq1$. When the scale of the protocol increases, we can adaptively adjust $\alpha$ with a larger value to counteract the effect by increasing $n$. In principle, this can bring a constant $SF$, which is the ideal result of this metric. (In practice, we may fine-tune $\alpha$ to better accommodate both a small $SF$, thus close to stably high throughput, and low latency, see $\S$\ref{section:experiment}). 
The new design paradigm not only provides an advantage in large-scale environments, but also makes it effective to further improve the efficiency of the system using the common method of scaling up. It enables an effective scaling up with $\gamma$ approximately to $\frac{1}{2}$ \textit{for all scales of the protocol}. This is a significant improvement and is substantially higher than the value $\gamma$ (approaches 0 as $n$ grows) in prior. 


\smallskip
\noindent
\textit{\underline{Further improving the protocol design.}} 
Amortizing the workload on disseminating requests to all replicas enables us to remove the bottleneck at the leader. 
It comes, however, at the risk of losing liveness in the presence of Byzantine faults: a faulty replica may always send its packages to a small subset of replicas and ignore others, thus  not enough votes can be collected to complete any confirmation. In the approach that using the leader to disseminate requests, liveness can be guaranteed by being equipped with a recovery mechanism (i.e., the view-change) to deal with a faulty leader. However, this is hard to recover from a faulty non-leader replica.
\footnote{The view-change mechanism monitors the progress of request confirmation, and switches faulty leaders if the progress stops. In the case that non-leader replicas disseminate requests, the stop of progress may be caused by faulty non-leader replicas rather than the leader. The switch of leaders cannot help to regain liveness in this case.}

We thus need an extra mechanism to restore  liveness. 
A trivial idea is to let replicas retrieve missing packages simply by asking the leader. 
However, this solution could eliminate the benefits brought by the workload balancing method, if a large number of packages should be re-sent by the leader. 
We instead let the replica retrieve from a committee of other replicas. The challenge is to ensure the existence of a committee that can indeed help (in the presence of Byzantine faults), 
meanwhile, no replica becomes a new bottleneck when $n$ is getting larger. We overcome this by adding one round of voting before initiating a new agreement instance, to ensure that a subset of replicas has received each outstanding package. The technique of erasure codes \cite{ECCtheory83} is also adopted. We show that a missing package will always be retrieved under an honest leader, and the amortized cost per replica still remains constant to $n$. 

\smallskip
\noindent\textbf{Experimental evaluations.} 
We implemented Leopard and conducted various experiments to evaluate its performance with up to 600 replicas. For a fair comparison, we deployed both Leopard and HotStuff back-to-back in the same environment. Significant and interesting results have been shown up, and the main results are as follows:

\begin{itemize}
	\item Leopard enables $5\times$ better throughput compared to HotStuff when $n$ is 300. The gap becomes wider with an even larger protocol scale. More importantly, the throughput of Leopard remains almost stable and achieves $10^5$ requests per second as the scale varies to 600. 
	\item Our evaluation shows that the throughput grows linearly with available bandwidth in both systems. However, under configurations with 20-200Mbps per-replica bandwidth, for example, Leopard achieves about 10-100Mbps throughput for all tested scales. Such throughput's growth is about 300\% of the throughput's growth in HotStuff when $n$ is 16 and 2000\% when $n$ is 128. 
	\item We evaluated the bandwidth utilization breakdown of our Leopard implementation across different components, and observed that, over 96\% bandwidth is utilized to disseminate requests. This implies that measuring the communication cost only on the vote phase is not enough for evaluating the efficiency of the protocol in the high throughput setting. 
\end{itemize}

We further carried out miscellaneous experiments, including the effect on varying batch sizes, leader's bandwidth utilization, the time spent at different components of the protocol, as well as the communication and time costs in the face of failures. The results obtained from our evaluations are inspiring (see $\S$\ref{section:experiment}), and let us have a better understanding of how Leopard behaves in many facets. It also motivates a further improvement on both the protocol and the implementation.

\section{Related Work}
\label{section:compare}

Reaching consensus in the face of Byzantine failures was first introduced by Lamport et al. in the 1980s \cite{PSL80reaching}. Early BFT protocols suffer from extremely high communication overhead \cite{DLS88}. The first practical BFT protocol in history is the seminal work of PBFT \cite{PBFT99}, and numerous improvements on leader-based BFT have been proposed since the invention of PBFT. 

In the past several years, cryptography, especially threshold signature schemes \cite{BLS01}, and new communication patterns are used to aggregate votes, allowing a single message acknowledgment rather than $O(n)$ messages for each decision. This effectively reduces the communication complexity by one order of $n$. 
Experimental evaluation on SBFT \cite{SBFT19} that adopts this method has shown that it is beneficial to a higher throughput. 
Besides, SBFT also introduces a fast path to reduce the voting round to one, in the case that no Byzantine replica exists. 

The chain-based idea, inspired by Bitcoin \cite{Nak08bitcoin}, has risen recently, in which each consensus proposal (e.g., a block) contains a hash link to its parent. By virtue of that, the confirmation of one proposal also means the confirmation of all its preceding proposals. When further combined with the elegant pipelining technique \cite{Casper17}, 
the total communication cost can be amortized and the average voting round can also be reduced to one. 
Experimental evaluation on HotStuff \cite{HotStuff-imple19}, which combines all the above ideas of aggregating votes and the chain-based pipelining technique, shows that it can achieve a peak throughput over $10^5$ requests per second. 

\newcolumntype{I}{!{\vrule width 1pt}}
\begin{table}[h]
	\vspace{-3mm}
	\centering
	\caption{Amortized cost when the leader is honest and after GST}
	\linespread{1.2}
	\vspace{-1mm}
	\small{
		\begin{threeparttable}
			\begin{tabular}{| m{2.2cm}<{\centering}|c|c|c|c|c|}
				\hline
				\multirow{2}{*}{Protocol}& \multicolumn{2}{c|}{Communication} & Scaling & \multicolumn{2}{c|}{\# of voting} \\
				\cline{2-3}\cline{5-6}
				& leader & non-leader & factor & optimistic & faulty\\
				\hline
				PBFT \cite{PBFT99} & \multirow{4}{*}{$O(n)$} & \multirow{4}{*}{$O(1)$} &\multirow{4}{*}{$O(n)$} & 2 & 2 \\
				\cline{1-1} \cline{5-6}
				SBFT \cite{SBFT19} & &&& 1 & 2 \\
				\cline{1-1} \cline{5-6}
				HotStuff \cite{HotStuff-imple19} (with pipelining) & & & & 1 & 1 \\
				\hline
				\textbf{Leopard}& $O(1)$& $O(1)$ & $O(1)$ &2 & 3 \\
				\hline
			\end{tabular}
		\end{threeparttable}
	}
	\vspace{-1mm}
	\label{table:compare}
\end{table}

Table \ref{table:compare} compares the amortized the communication complexity, the scaling factor and number of voting rounds, among our protocol and some of the most prominent prior works, 
when the leader is honest and after GST, i.e., the most possible case a high throughput can be achieved. We compare the number of voting rounds in two cases: the \textit{optimistic} case that all replicas follow the protocol, and the \textit{faulty} case that there are $f$ faulty non-leader replicas who may arbitrarily deviate from the protocol. 
We see that the leader's amortized communication complexity and the scaling factor in all compared protocols are $O(n)$, whereas Leopard achieves $O(1)$ on both these measures. This benefits largely from the technique of decoupling. 

The technique of decoupling is not new and has been used, for example, in Prism \cite{Prism19}, which is an excellent blockchain protocol improving the efficiency of Bitcoin. Our protocol is inspired by it, but with a different purpose: we decouple blocks to balance the workload of replicas and remove the bottleneck in leader-based BFT protocols. This technique also inspires other works. Recently, we were made aware of a concurrent work \cite{Narwhal21} 
using the reliable broadcast (RBC), a common technique to balance the workload of a specific node on data dissemination. While our protocol aims at preserving throughput when the protocol's scale increases, \cite{Narwhal21} targets designing a BFT protocol with higher throughput. Besides, the quadratic communication complexity in existing RBC schemes \cite{ADD21} is not friendly to support large scale replicas. 
Every block of requests in \cite{Narwhal21} requires to link $2f+1$ blocks in the previous round from different replicas. In the presence of bandwidth variability at replicas, the speed of request dissemination in \cite{Narwhal21} is thus subject to the slowest $f+1$ replica. As a contrast, the request dissemination in our protocol proceeds independently at each replica, thus it benefits to take full advantage of the bandwidth resource of replicas.  

There are several other works 
also mention the bottleneck in leader-based BFT protocols. However, some of them bypass the problem and move to the leader-less BFT paradigm \cite{Dumbo20}. 
Another recent work of Kauri \cite{Kauri21} addresses the bottleneck by leveraging the communication tree 
to optimize data dissemination. 
Although being improved by Kauri, the advantage of this technique is quite fragile to faults, where a faulty replica can stop delivering the message to replicas residing on the subtree rooted by it. Since there are about 1/3 fraction of the replicas are faulty, this technique seems still precarious.

\section{Preliminaries}
\label{section:model}

\subsection{Model}
We consider a system consisting of a fixed set of $n=3f+1$ replicas, indexed by $i\in[n]$, where at most $f$ replicas are Byzantine and other replicas are honest. 
We often refer to Byzantine replicas as being coordinated and fully controlled by an adversary. 
Each replica holds a signature key pair $(tpk_i,tsk_i)$ and the master public key $mpk$ of the threshold signature scheme. The identities of replicas and their public keys are known to all. 
The adversary cannot break the cryptographic primitives. 
Network communication among replicas is point-to-point, authenticated, and reliable: an honest replica receives a message from another replica if and only if the latter sent that to the former. We adopt the\textit{ partially synchronous network model} of Dwork et al. \cite{DLS88}, where a known bound $\Delta$ on message transmission holds 
after some unknown time point called the global stabilization time (GST). 
This model follows many BFT protocols, especially PBFT and HotStuff. 

\subsection{Notations}
We make standard cryptographic assumptions and we use the following cryptographic constructions. $\mathsf{H}(m)$ is a collision-resistant cryptographic hash function, producing a hash value for the input message $m$. 


We make use of a $(2f+1,n)$-threshold signature scheme $\mathcal{TS}=(\mathsf{TSig,TVrf,TSR})$: $\mathsf{TSig}(tsk_i,m)$ is the threshold signature algorithm on inputting secret key $tsk_i$ of replica $i$ and message $m$, it outputs a signature share $\hat{\sigma}_i$; 
$\mathsf{TVrf}(tpk_i,\hat{\sigma}_i,m)$ is the verification algorithm that outputs 1 if $\hat{\sigma}_i$ is a valid threshold signature share of $m$ under $tpk_i$, and outpus 0 otherwise; 
$\mathsf{TSR}(S)$ is the combining algorithm on inputting a set $S$ of $2f+1$ valid signature shares of message $m$, it outputs a combined signature $\hat{\sigma}$ which is verified by invoking $\mathsf{TVrf}(tpk,\hat{\sigma},m)$.

We use $(f+1, n)$-erasure codes \cite{ECCtheory83} where a message $m$ is encoded into a set of $n$ chunks, of which any subset of $f+1$ valid chunks will decode to obtain $m$.

\subsection{Design goals} 
We aim to design a leader-based BFT protocol that preserves high throughput when the protocol's scale increases. The security and efficiency goals are as follows: 

\smallskip
\noindent\textbf{The security goals.} A secure BFT protocol should satisfy two basic security properties:
\begin{itemize}
	\item \textit{Safety}: The requests recorded in the same position of two output logs from honest replicas are identical.
	\item \textit{Liveness}: Every request submitted by the client eventually receives an acknowledgment of confirmation.
\end{itemize}

\noindent\textbf{The efficiency goal.} The efficiency goal is measured by a new performance metric called \textit{scaling factor}, which captures the relation between the scale of the protocol and the maximal communication cost among replicas.

\begin{definition}[Scaling Factor]
	Scaling factor of a BFT protocol is the heaviest communication cost in bits, with respect to the scale of the protocol, among all honest replicas on confirming one bit of request by executing the protocol.
\label{def:SFC}
\end{definition}

To preserve the throughput as the protocol's scale increases, it should satisfy that\footnote{We also expect the increase of protocol's throughput is 
	significant to added bandwidth resource when scaling up the system. We omit this as an efficiency goal here, since to our knowledge, no existing tight upper bound defining the fraction of resource usage in a leader-based BFT protocol with the ``best" scaling up performance. We leave it as future work. Nevertheless, we present an analysis of the effectiveness when scaling up Leopard (shown in $\S$\ref{section:analysis}).}:



%

\begin{itemize}
	\item \textit{Constant scaling factor}: After GST, the scaling factor of the protocol is constant.	
	

\end{itemize}

\section{Protocol of Leopard}
\label{section:protocol}

\begin{figure} 
	\centering{\includegraphics[width=6cm]{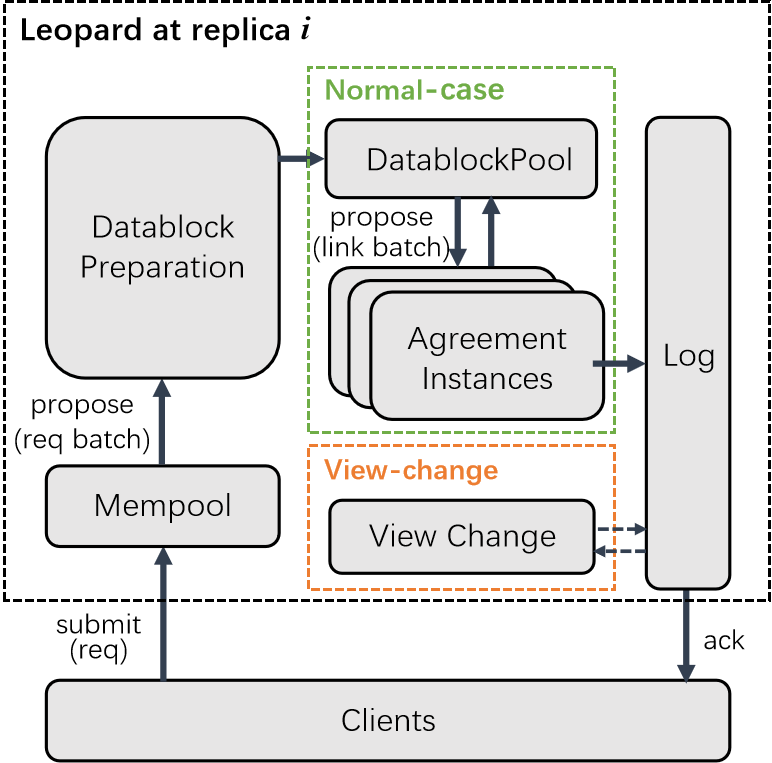}}
	\centering \caption{The architecture of Leopard.}
	\vspace{-4mm}
	\label{fig:architecture}
\end{figure}

The protocol proceeds in successive views starting from 1. Each view consists of multiple agreement instances and a view-change. An agreement instance is to confirm a consensus proposal, proposed by the view leader which is known to all. A view-change is to replace a faulty leader and advance the protocol to the next view. 
The output of our protocol is a growing log of confirmed requests from agreement instances. 
Fig. \ref{fig:architecture} depicts the architecture of Leopard that features a memory, a datablock preparation, a normal-case modular, 
a view-change modular 
and a log. 
The typical message flow of Leopard is shown in Fig. \ref{fig:message-flow}. 
\begin{figure} [!h]
	\vspace{-2mm}
	\centering{\includegraphics[width=8.5cm]{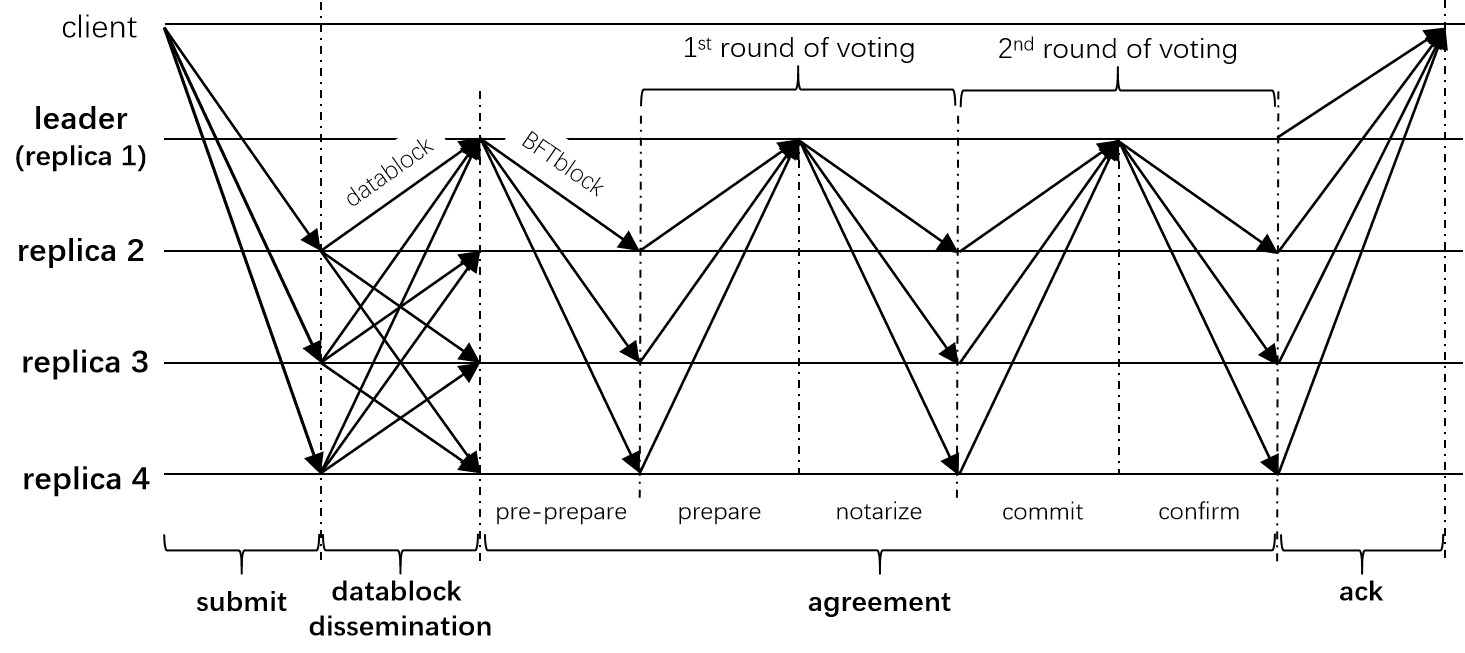}}
	\vspace{-2mm}
	\centering \caption{The typical message flow of Leopard with $n=4$.}
	\label{fig:message-flow}
\end{figure}

Before describing the detail of the protocol, we first define two newly introduced blocks: \textit{datablock} and \textit{BFTblock}. 

(i) Datablock: A datablock is generated by a non-leader replica using requests received from clients. The formation of a datablock is: $\langle\texttt{datablock}, (i,counter), \mathcal{R}\rangle$, where $\mathcal{R}$ is a set of requests, $i$ is the identity of its generator, and $counter$ shows the number of datablocks having been generated by $i$. 

(ii) BFTblock: The BFTblock is what replicas agree on. 
It is generated by the current leader, using received outstanding datablocks. A BFTblock only contains links (i.e., hashes) to the involved datablocks. The formation of a BFTblock is: $\langle\texttt{BFTblock},(v,sn),ct\rangle$, where $v$ is the view number, $sn$ is the serial number assigned by the leader, and $ct$ is the content (i.e., hashes of datablocks). 
A BFTblock has two states: \textit{notarized} and \textit{confirmed}. A BFTblock is notarized if it has completed the first round of voting, thus a corresponding \textit{notarization proof} exists; a BFTblock is confirmed if it has completed the second round of voting, thus a corresponding \textit{confirmation proof} exists. Replicas change the state of BFTblocks on receiving valid proofs.

\subsection{Protocol Description}
\label{sec:protocol-description}
Leopard is composed of a suite of tiered components to reach consensus on pending requests. For most of the time, replicas execute the normal-case modular to confirm requests, a.k.a. the normal-case mode of the protocol. This happens as long as the current view leader is honest and after GST. We present the protocol description, including datablock preparation and agreement, in this section. These are components of the normal-case modular.   
The view-change modular is triggered only if the progress of confirmation delays. It moves the protocol to the next view with a new view leader, and then starts the normal-case mode all over again. 
Our view-change mechanism for dealing with faulty leaders is mostly based on PBFT, which we defer the details of this part to appendix \ref{appendix:detailprotocol}. In the below, we assume the current view number is $v$ and denote the view leader as $L_v$.

\subsubsection{Datablock Preparation}
Datablock preparation is to synchronize pending requests among replicas by non-leader replicas, who receive requests from clients. 
A client identifies a non-leader replica out of the available replicas and submits its pending requests to the replica. It then waits for acknowledgments of their confirmations. 


In the large-scale environment, geo-distributed replicas are likely to receive requests from their neighbor clients, which makes the received requests by a replica disjoint from other replicas. Considering that the request generation rate at different regions may vary, some replicas may be too busy to process a large volume of requests, while others may be idle. To balance the processing volume among replicas, a request $req$ can also be uniformly assigned by executing a deterministic function $\mu(req)$, which returns the identifier of the replica who is responsible to process $req$. Note that this determinism does not imply predictability as one can change the responsible replicas directly \cite{redbelly21}. 
To prevent a Byzantine replica who may censor (e.g., drop or ignore) some requests, a client could change another responsible replica and re-submit the request to it, if no acknowledgment is received after a timer is expired (as in \cite{SBFT19,PBFT99}). The timer can be adaptively set based on past network profiling in practice. Since there are at most $f$ Byzantine replicas, up to $f$ times changes will guarantee the existence of an honest replica, who is centain to disseminate the request to all the others. The number of identified replicas in each submit can also be as large as $f+1$, however, more replicas lower latency whereas fewer replicas increase throughput. 

On receiving requests from the client, non-leader replicas generate datablocks following Algorithm \ref{alg:datablock}. This process occurs on an ongoing basis and proceeds independently at each replica. Replicas will \textit{continually} process pending requests to take full advantage of resources. 
Specifically, each non-leader replica records a counter $d$ initiated by 1. It extracts a set $\mathcal{R}$ of pending requests, and a new datablock $m$ is generated using $\mathcal{R}$ and the current counter $d$. After $m$ has been multicast, $d$ is increased by 1 indicating that a new datablock has been generated\footnote{The counter can also help to limit the number of valid datablocks during a period of time. This rate-limit mechanism is a common technique to avoid the flooding attack \cite{rate-limit14} by Byzantine replicas, who may send a massive amount of datablocks filled with repetitive requests.}. 
On receiving a datablock $m$ from replica $i$, replica $j$ puts it into its datablockPool, if no datablock with a repetitive counter has been received from $i$.

\begin{algorithm}
	\caption{Datablock Dissemination}
	\small{
		\begin{algorithmic}[1] 
			\State \textcolor{purple}{// Datablock generation}
			\State (\textbf{for} each non-leader replica $i,i\in[n]$)
			\State Init: $d\leftarrow1$
			\While{()}
			\State  let $\mathcal{R}$ be a set of requests extracted from its mempool
			\State $m\gets\langle\texttt{datablock}, (i,d),\mathcal{R}\rangle$
			\State  multicast $m$ to all other replicas 
			\State  $d\gets d+1$
			\State mempool $\leftarrow$ mempool$\backslash\mathcal{R}$ \cmd{remove processed requests}
			\EndWhile
			
			\vspace{1mm}
			\State \textcolor{purple}{// Datablock verification}
			\State (\textbf{for} replica $j,j\in[n]$)
			\State \textbf{upon} receiving a datablock $m$ from replica $i$
			\State\quad\textbf{if} no datablock with the same $d$ has been received from $i$ \textbf{then}
			\State\quad\quad datablockPool $\leftarrow\textrm{datablockPool}\cup\{m\}$  
			\State\quad\quad\textbf{return} accept
		\end{algorithmic}
	}
	\label{alg:datablock}
\end{algorithm}

\subsubsection{Agreeing on BFTblocks}
Requests are being confirmed by a sequence of agreement instances, each processes a BFTblock generated by the leader $L_v$, and these instances can be conducted \textit{in parallel}. This enables handling new requests without waiting for the completion of previous ones. We limit the number of parallel-executed BFTblocks to $k$ (e.g., 100 as in PBFT), to avoid a Byzantine leader picks a number that is so high to exhaust the space of serial numbers. We also set a watermark $lw$ to limit the region of a valid serial number of BFTblock. Since $k$ can be set large enough and the parameter $lw$ will be advanced by the garbage collection mechanism described in appendix \ref{appendix:detailprotocol}, the process of agreeing will not stall waiting for the completion of previous BFTblocks. Thus, the parallel execution can be maintained.  

An agreement instance consists of two-round voting, initiated by the leader proposing a BFTblock. 
We adopt threshold signatures to reduce the communication cost in each agreement instance. Due to that, each voting round consists of two stages: one is for replicas sending votes to a specific node, another is to multicast a completion proof of this voting round by the node. 
Each vote is instantiated by a threshold signature share, and a proof is derived by aggregating $(2f+1)$ votes. The specific node is critical for the completion of the agreement, and we let the leader do this job as in \cite{HotStuff-imple19} since a faulty leader will be replaced anyway (by the view-change). 

\begin{algorithm}[!t]
	\caption{Agreeing on BFTblocks}
	\small{
		\begin{algorithmic}[1] 
			\State Init: $sn\gets 1;lw\leftarrow0$
			\vspace{1mm}
			\State \textcolor{purple}{// Pre-prepare stage} 
			\State (\textbf{for} leader $L_v$)
			\State \quad let $ct$ be a set containing the hashes of outstanding datablocks
			\State \quad$m \gets\langle\texttt{BFTblock},(v,sn),ct\rangle$ \cmd{generate a BFTblock $m$}
			\State \quad$\hat{\sigma}_{L_v}^1\gets\mathsf{TSig}(tsk_{L_v},\mathsf{H}(m))$ \cmd{sign the hash of $m$ with $tsk_{L_v}$}
			\State \quad multicast $(m,\hat{\sigma}_{L_v}^1)$ to all replicas
			\State \quad$sn\gets sn+1$
			
			\vspace{1mm}
			\State \textcolor{purple}{// Prepare stage}
			\State (\textbf{for} replica $i,i\in[n]$)
			\If {$\textsc{VrfBFTblock}(m,\hat{\sigma}_{L_v}^1)\rightarrow1$} 			
			\State $\hat{\sigma}_{i}^1\gets\mathsf{TSig}(tsk_{i},\mathsf{H}(m))$\cmd{generate $i$'s signature share}
			\State send $(\mathsf{H}(m),\hat{\sigma}_{i}^1)$ to $L_v$\cmd{first-round vote}
			\EndIf
			
			\vspace{1mm}
			\State \textcolor{purple}{// Notarize stage}
			\State (\textbf{for} leader $L_v$)
			\State \quad wait for $2f+1$ valid signature shares for $\mathsf{H}(m)$
			\State\quad $\hat{\sigma}^1\gets\mathsf{TSR}(\{\hat{\sigma}_j^1\}_{j=1}^{2f+1})$ \cmd{the notarization proof for $m$}
			\State \quad multicast $(\mathsf{H}(m),\hat{\sigma}^1)$ to all replicas
			
			\vspace{1mm}
			\State \textcolor{purple}{// Commit stage}
			\State (\textbf{for} replica $i,i\in[n]$)
			\If {$\mathsf{TVrf}(tpk,\hat{\sigma}^1,\mathsf{H}(m))\rightarrow 1$}
			\State change the state of $m$ to \textit{notarized} 
			\State $\hat{\sigma}_{i}^2\gets\mathsf{TSig}(tsk_{i},\mathsf{H}(\hat{\sigma}^1))$ 
			\State send $(\mathsf{H}(\hat{\sigma}^1),\hat{\sigma}_{i}^2)$ to $L_v$ \cmd{second-round vote}
			\EndIf
			
			\vspace{1mm}
			\State \textcolor{purple}{// Confirm stage}
			\State (\textbf{for} leader $L_v$)
			\State \quad wait for $2f+1$ valid signature shares for $\mathsf{H}(\hat{\sigma}^1)$
			\State \quad $\hat{\sigma}^2\gets\mathsf{TSR}(\{\hat{\sigma}_j^2\}_{j=1}^{2f+1})$ \cmd{the confirmation proof for $m$}
			\State \quad multicast $(\mathsf{H}(\hat{\sigma}^1),\hat{\sigma}^2)$ to all replicas
			
			\vspace{1mm}
			\State (\textbf{for} replica $i,i\in[n]$)
			\State\textbf{if} {$\mathsf{TVrf}(tpk,\hat{\sigma}^2,\mathsf{H}(\hat{\sigma}^1))\rightarrow 1$} \textbf{then}
			\State \quad change the state of $m$ to \textit{confirmed}
			\State \quad add $m$ into $i$'s log following the order of $m$'s serial number
			
			\vspace{3mm}
			\Function{VrfBFTblock}{$m,\hat{\sigma}_{L_v}^1$}
			\State\textbf{if} $(\mathsf{TVrf}(tpk_{L_v},\hat{\sigma}_{L_v}^1,\mathsf{H}(m))\rightarrow0) || (m.v$ is not the current view number$) || ($some other BFTblock with the same $sn$ has been voted in view $v)|| (lw<sn\leq lw+k)$ ~\textbf{then} \cmd{similar to the condition in PBFT}
			\State\qquad\textbf{return} 0			
			\State\textbf{else if} $\exists\mu\in m.ct,$ s.t. no corresponding datablock in datablockPool ~\textbf{then} \cmd{the datablock has not been received}
			\State\qquad retrieve every missing datablock \cmd{invoke the datablock retrieval mechanism}
			\State\textbf{return} 1 \cmd{all linked datablocks in $m$ have been received}
			\EndFunction		
		\end{algorithmic}
	}
	\label{alg:normal-case}
\end{algorithm}

%

As shown in Algorithm \ref{alg:normal-case}, each consensus proposal (i.e., BFTblock) only contains links to datablocks in the leader's datablockPool. 
It thus requires to ensure the receival of all linked datablocks before casting a vote on the BFTblock. Specifically, each replica checks the receival of every linked datablock (line 39) on receiving a BFTblock, using the \textsc{VrfBFTblock} function. Only if all linked datablocks have been received, the replica can cast the first-round vote on this BFTblock. 

\smallskip
\noindent\textbf{Datablock Retrieval.} 
The confirmation of requests works properly if replicas follow the protocol. However, in the presence of Byzantine faults, our data-decoupled BFT paradigm may cause that, a faulty non-leader replica only sends its generated datablocks to a small subset of replicas 
and ignores others, thus conducting a selective attack during the datablock dissemination. Since it's hard to identify datablocks coming from faulty replicas, the leader may multicast BFTblocks containing datablocks that will never be received by other replicas. This leads to not enough votes for completing the confirmation, thus breaking the liveness.

To preclude this problem, a recovery mechanism is needed to restore liveness. 
Notice that replicas discover a missing datablock only after receiving a BFTblock from the leader. When the leader is honest, it must have received each datablock linked by its generated  BFTblock. An intuitive solution is to let the leader help recover. 
A replica sends a query 
to the leader, who will respond with the corresponding datablock, if it is the first time of query for this datablock from this replica. 

\begin{figure}[t]
	\vspace{-2mm}
	\centering{\includegraphics[width=8.5cm]{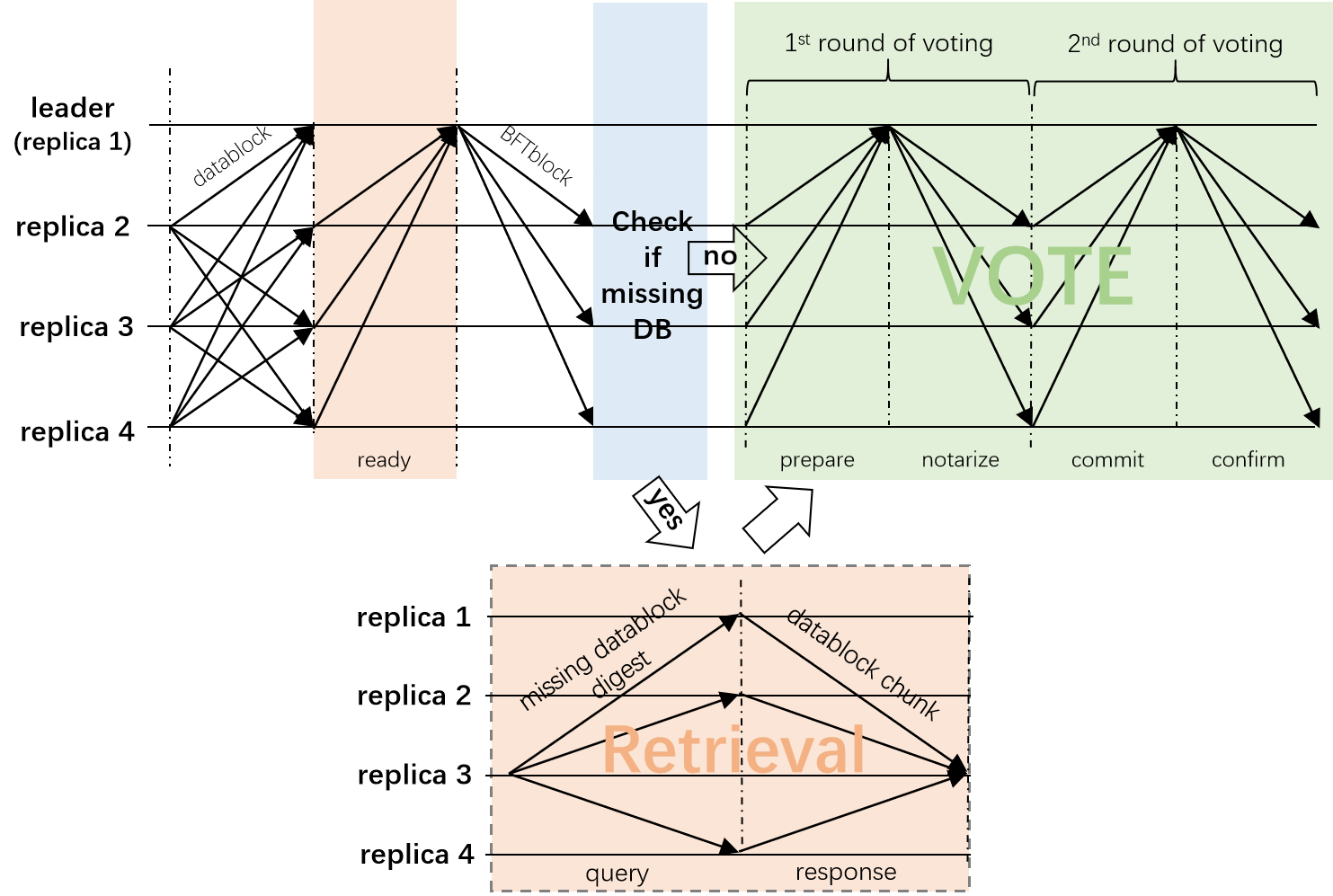}}
	\vspace{-2mm}
	\centering \caption{The message flow of Leopard with the improvement on the retrieval mechanism by replica $i$, using $i=3$ as an example.}
	\vspace{-6mm}
	\label{fig:retrieve}
\end{figure}

The intuitive solution works in most cases. However, it may overload the leader, if a large number of datablocks should be retrieved. This can happen since about 1/3 fraction of the replicas can be faulty, and each faulty replica may send all its datablocks only to the leader. Even worse, a faulty replica may pretend to have not received datablocks from honest replicas, and query these datablocks to further overload the leader. Hence, this solution may eliminate the benefits brought by our workload balancing method in the above.


To improve the retrieval, we instead use a committee of replicas to collectively help retrieve missing datablocks. 
The algorithm is shown in Algorithm \ref{alg:retrieve2}, and the message flow is shown in Fig. \ref{fig:retrieve}. A replica checks whether there is a missing datablock when receiving a BFTblock. If no missing datablock appears, the replica starts the voting period on the BFTblock; otherwise, the replica retrieves from a committee of other replicas. After that, the replica gets back to the voting period. 
To ensure the existence of a committee of replicas that have received the queried datablocks and are willing to help, we add an extra round of voting before the leader initiates a new agreement instance. On accepting new datablocks, replicas send ready messages containing the hashes of datablocks to the leader. 
The leader waits for at least $2f+1$ ready messages for a datablock before linking it by the generated BFTblock. 

Once a replica $i$ realizes a missing datablock, it multicasts a query message. If replica $j$ has the datablock and it is the first time receiving a query of this datablock from $i$, $j$ 
encodes the datablock using a $(f+1,n)$-erasure code and constructs a Merkle tree out of the encoded chunks. $j$ sends the chunk $C_j$ along with the Merkle root $rt$ and a Merkle tree proof $proof_j$ that proves the chunk $C_j$ belongs to root $rt$. Replica $i$ needs to wait for at least $f+1$ chunks under the same Merkle root for decoding and obtaining the datablock. 
When multiple missing datablocks are found, a query can be sent along with all hashes of these datablocks.

\begin{algorithm}
	\caption{Retrieval}
	\small{
		\begin{algorithmic}[1]
			\State \textcolor{purple}{// Ready}
			\State (\textbf{for} replica $i,i\in[n]$)
			\State \textbf{upon} accepting a datablock $m$
			\State\qquad send $\langle\texttt{Ready},\mathsf{H}(m),j\rangle$ to leader $L_v$
			
			\vspace{1mm}
			\State (\textbf{for} leader $L_v$)
			\State \textbf{upon} receiving $2f+1$ $\langle\texttt{Ready},\mathsf{H}(m),\_\rangle$ messages from distinct replicas for a datablock $m$ in datablockPool
			\State\qquad move $m$ to $L_v$'s readyblockPool \cmd{BFTblocks are required to link datablocks from the readyblockPool}
			
			\vspace{1mm}
			\State \textcolor{purple}{// Query}
			\State (\textbf{for} replica $i$)
			\State Timer.Start() \cmd{start a timer on the missing datablock}
			\State\textbf{if} receiving a valid datablock with digest $\mu$ \textbf{then}
			\State\quad Timer.Stop() \cmd{no query will be triggered on this datablock}
			\State\textbf{else if} Timer.Timeout() \textbf{then}
			\State\quad $\hat{m}\gets\langle\texttt{Query},(\mu,i)\rangle$
			\State\quad multicast $\hat{m}$ to all replicas
			
			\vspace{1mm}
			\State \textcolor{purple}{// Response}
			\State (\textbf{for} replica $j,j\in[n]$)
			\State\textbf{if} $(\exists m\in$ datablockPool, s.t., $\mathsf{H}(m)=\mu)$ \& ($m$ has not been queried by replica $i$) ~\textbf{then} \cmd{find the corresponding datablock}
			\State\qquad encode the datablock $m$ using a $(f+1,n)$-erasure code, and output $n$ chunks $C_1,C_2,\cdots,C_n$
			\State\qquad compute the Merkle tree root $rt$ with all chunks as leaves
			\State\qquad send $\langle\texttt{Resp},rt,C_j,\textrm{proof}_i,j\rangle$ to replica $i$, where $\textrm{proof}_j$ is the Merkle tree proof for $C_j$ under root $rt$
			
			\vspace{1mm}
			\State (\textbf{for} replica $i$)
			\State\textbf{upon} receiving $\langle\texttt{Resp},rt,C_j,\textrm{proof}_j,j\rangle$ from replica $j$
			\State\quad\textbf{if} $C_j$ is valid under root $rt$ by verifying $\textrm{proof}_j$ ~\textbf{then} 
			\State\qquad store the chunk $C_j$ with the root $rt$
			\State\quad\textbf{end if}
			\State\quad wait for $f+1$ valid chunks with the same root $rt$
			\State\quad decode using these chunks to recover datablock $m$
		\end{algorithmic}
	}
	\label{alg:retrieve2}
\end{algorithm}

\noindent\textbf{Execution and acknowledgement.} 
Since a BFTblock contains links to datablocks and a datablock contains requests, the confirmation of a BFTblock means the confirmation of all its linked requests. Replicas store confirmed BFTblocks in its local log following the order of the serial number. The corresponding requests are also logically stored in the log, whereas the requests linked by a BFTblock can be ordered in alphabetical order. Replicas execute requests linked in the longest sequence of consecutive BFTblocks of their local logs, after which acknowledgements will be returned to clients.

We remark that this paper focuses on reaching consensus among replicas, as in many other BFT protocols \cite{PBFT99,HotStuff-imple19}, and hence the current version doesn't check the request validity. When it is being applied, an application-specific predicate $\mathsf{verify}(\cdot)$ should be invoked on each request before casting the first round of voting. Since the protocol requires that each linked datablocks by BFTblocks has been received before casting a vote, $\mathsf{verify}(\cdot)$ can be an extra condition in $\textsc{VrfBFTblock}(\cdot)$. Also, a similar technique as in \cite{redbelly21} can be used to reconcile conflicted requests before they output. 

\section{Analysis}
\label{section:analysis}
\subsection{Security Analysis}
The safety of Leopard is captured in Theorem \ref{theo:safety1}, and the liveness is captured in Theorem \ref{theo:liveness1}. Due to the FLP impossibility \cite{FLP83}, the liveness of Leopard relies on the synchronous assumption. We present a proof sketch below. The full version of the proof appears in appendix \ref{appendix:security}.

\begin{theorem}[Safety]
	If two different requests exist on the same position of two output logs, then the two logs cannot be both from honest replicas.
	\label{theo:safety1}
\end{theorem}

\textit{Proof sketch:} Since the confirmations of requests are reached through the confirmations of BFTblocks, we only have to prove that no two different BFTblocks with the same serial number appear in two output logs of honest replicas. The proof follows that of PBFT. The idea is that, a BFTblock gets confirmed only if $2f+1$ replicas have voted on it, which contain at least $f+1$ honest replicas. Since there are $2f+1$ honest replicas in total, if two different BFTblocks with the same serial number are confirmed, there must be at least one honest replica who has voted on both of them, which comes to a contradiction according to the protocol. 

\begin{theorem}[Liveness]
	After GST, the confirmation for a pending request will always be reached. 
	\label{theo:liveness1}
\end{theorem}

\textit{Proof sketch:} 
In the case that the leader is honest, it will link every outstanding datablock that has been received in its generated BFTblocks, and initiate new agreement instances to confirm these BFTblocks. A pending request will be delivered to the leader through some datablock, otherwise, the client will resend the request to at most $f$ replicas, in which at least one honest replica will multicast the request in its datablock to all the other replicas including the leader. If a datablock $m$ is linked by a BFTblock, according to our datablock retrieval mechanism (see Algorithm \ref{alg:retrieve2}), there must be at least $2f+1$ ready messages for $m$, and it contains at least $f+1$ honest replicas who have obtained $m$. 
Hence, there must be enough (i.e., $f+1$) responses to help an honest replica who missed $m$ to recover $m$. Since there are at least $2f+1$ honest replicas, the confirmations of the honest leader's BFTblocks will be reached. 

In the case that the current leader is faulty, the liveness property is guaranteed by the view-change mechanism. The proof of this case follows that in PBFT.

\subsection{Efficiency Analysis}

We analyze the efficiency of Leopard using the performance metric of scaling factor, defined in $\S$\ref{section:model}. We also analyze the effectiveness of scaling up. Firstly, we discuss the communication cost of Leopard in three cases:

(a) In the optimistic case that all replicas follow the protocol and the network is during the period of synchronous, only Algorithm \ref{alg:datablock} and \ref{alg:normal-case} will be executed.  
Since the request dissemination is amortized by $n-1$ non-leader replicas, let $\Lambda_r$ be the number of pending requests, \textit{payload} denote the size of a request, and hence, each non-leader replica will process $\frac{\Lambda_r\times\textrm{payload}}{n-1}$ bits of requests. 

Let $\alpha$ be the number of bits per datablock, $\beta$ be the size of a hash output, and $\kappa$ be the size of a vote (i.e., a threshold signature). When processing $\Lambda_r$ requests, it will generate $\frac{\Lambda_r\times\textrm{payload}}{\alpha}$ datablocks in total. Since a BFTblock only contains the hashes of datablocks, we have that the total bits of generated BFTblocks are $\frac{\Lambda_r\times\textrm{payload}}{\alpha}\times\beta$. Let $\tau$ be the batch size of a BFTblock, the number of BFTblock, i.e., the number of agreement instances for processing $\Lambda_r$ requests, is $\frac{\Lambda_r\times\textrm{payload}}{\alpha\tau}$.

Following the protocol, the communication cost $\overline{c}_L$ of the leader includes the costs of receiving datablocks, multicasting BFTblocks, and processing votes. Thus, we have 
\begin{equation}
	\begin{aligned}
		\overline{c}_L&=\overline{c}_L^\textrm{recvDB}+\overline{c}_L^\textrm{agreement}\\
		&=\frac{\Lambda_r\times\textrm{payload}}{n-1}\times(n-1)\\
		&\hspace{4mm}+\frac{\Lambda_r\times\textrm{payload}}{\alpha}\cdot(\beta+\frac{1}{\tau}4\kappa)\cdot(n-1)\\
		&=(\frac{(\beta+4\kappa/\tau)\cdot(n-1)}{\alpha}+1)\cdot(\Lambda_r\times\textrm{payload}).
	\end{aligned}
	\label{eq:opt-leader}
\end{equation}

As a non-leader replica, it receives pending requests from clients and datablocks from the other $n-2$ replicas. It multicasts its generated datablocks to all the other $n-1$ replicas. It also deals with BFTblocks and votes during agreement instances. Hence, a non-leader replica's communication cost $\overline{c}_R$ on processing $\Lambda_r$ requests is,
\begin{equation}
	\begin{aligned}
		\overline{c}_R&=\overline{c}_R^\textrm{DB}+\overline{c}_R^\textrm{agreement}\\
		&=\frac{\Lambda_r\times\textrm{payload}}{n-1}\times(1+n-2+n-1) \\
		&\hspace{4mm}+\frac{\Lambda_r\times\textrm{payload}}{\alpha}\cdot(\beta+\frac{1}{\tau}4\kappa)\\
		&=(2+\frac{\beta+4\kappa/\tau}{\alpha})\cdot(\Lambda_r\times\textrm{payload}).
	\end{aligned}
	\label{eq:opt-member}
\end{equation}

(b) In the case that the network is during the period of synchronous and the leader is honest, but there are $f$ faulty replicas conducting the selective attack, such that their datablocks are only multicast to a subset of $(n-s)$ replicas including the leader, and they drop all datablocks from honest replicas, the retrieve algorithm will also be executed. 
Since $n=3f+1$, about $1/3$ fraction of datablocks are generated by faulty replicas, and each one of them will be queried at most $s$ times;  
about $2/3$ fraction of datablocks are generated by honest replicas, and each of them will be queried at most $f$ times.  
The total number of datablocks need to respond can be estimated by $\frac{\Lambda_r\times\textrm{payload}}{\alpha}\cdot(\frac{1}{3}s+\frac{2}{3}f)$. 
Since $s\leq 3f$, we have that the total number of responses is at most $\frac{\Lambda_r\times\textrm{payload}}{\alpha}\cdot\frac{5f}{3}$.

When using the intuitive solution to retrieve datablocks that all missing datablocks are responded to by the leader, the cost on responding is  $\frac{5}{3}(\Lambda_r\times\textrm{payload})\cdot f$. Thus, the communication complexity of the leader is $O(n)$. 
While for the improved retrieval mechanism (in Algorithm \ref{alg:retrieve2}), the size of each response is $(\frac{\alpha}{f+1}+\beta\log n)$. For each replica, the cost on responding is at most $(\frac{\alpha}{f+1}+\beta\log n)\times\frac{\Lambda_r\times\textrm{payload}}{\alpha}\frac{5f}{3}$. As an honest replica, it only queries datablocks generated by faulty replicas. To recover a missing datablock, it needs to receive $(f+1)$ responses from different replicas. 
Hence, the cost on receiving missing datablocks is $(f+1)(\frac{\alpha}{f+1}+\beta\log n)\cdot\frac{\Lambda_r\times\textrm{payload}}{3\alpha}$. Since replicas also have to send ready messages for each datablock to the leader, the extra communication cost per replica is at most $\frac{5\Lambda_r\times\textrm{payload}}{3\alpha}\cdot(\alpha+\beta(f\log n+\frac{3}{5}))$.


(c) In the worst case that the network is during the period of asynchronous, the analysis of the cost on datablock retrieval is similar to case (b) with a difference that, 
even an honest replica may query every datablock. According to case (b), we deduce that the cost on receiving missing datablocks per honest replica is at most $(f+1)(\frac{\alpha}{f+1}+\beta\log n)\cdot\frac{\Lambda_r\times\textrm{payload}}{\alpha}$, whereas the maximal cost on responding remains at $(\frac{\alpha}{f+1}+\beta\log n)\times\frac{\Lambda_r\times\textrm{payload}}{\alpha}\frac{5f}{3}$. Therefore, the cost per honest replica on datablock retrieval is no more than $\frac{5\Lambda_r\times\textrm{payload}}{3\alpha}\cdot(\alpha+\beta((f+1)\log n+\frac{3}{5}))$.

\vspace{1mm}
Now, we analyze the scaling factor. From \eqref{eq:opt-leader} and \eqref{eq:opt-member}, we have that

\begin{equation*}
	\begin{aligned}
		SF&=\max\{\overline{c}_L,\overline{c}_R\}/(\Lambda_r\times\textrm{payload}) \\
		&=\max\{\frac{(\beta+4\kappa/\tau)\cdot(n-1)}{\alpha}+1,2+\frac{\beta+4\kappa/\tau}{\alpha}\}.
		\label{eq:SF-ours}
	\end{aligned}
\end{equation*}
Notice that, $\alpha$ is a system parameter denoting the datablock size. When setting $\alpha=\lambda(n-1)$ where $\lambda$ is a constant positive number, we can obtain a scaling factor of $O(1)$. 

Next, we give an analysis of the effectiveness of scaling up. When adding more bandwidth resources to each replica, the processing capacity of the replica is increased. Let $\Lambda_b^\Delta$ denote the increased bits of requests that can be processed, and $C^\Delta$ denote the increased effective communication capacity when executing the protocol at each replica. Thus, we have
\begin{equation}
	\begin{aligned}
		\frac{\Lambda_b^\Delta}{C^\Delta}&=\frac{\Lambda_b^\Delta}{\max\{\overline{c}_L^\Delta,\overline{c}_R^\Delta\}}\\
		&=\frac{1}{\max\{\frac{(\beta+4\kappa/\tau)\cdot(n-1)}{\alpha}+1,2+\frac{\beta+4\kappa/\tau}{\alpha}\}}.
	\end{aligned}
	\label{eq:ce_ours}
\end{equation} 
When $\beta+4\kappa/\tau\leq\lambda$,\footnote{We note that $\beta+4\kappa/\tau\leq\lambda$ can be generally satisfied in implementations. Since the parameter $\alpha$ denotes the size of a datablock and $\alpha=\lambda(n-1)$, we have that the inequation equals to $n\leq\frac{\textrm{payload}}{\beta+4\kappa/\tau}X+1$ where $X$ is the number of requests in a datablock. We adopt the parameters like other BFT implementations \cite{HotStuff-imple19,SBFT19,BFTSmart} in our evaluation: $\beta=32$ bytes (use hash functions based on SHA-256), $\kappa=48$ bytes (use threshold signatures based on threshold BLS \cite{BLS01}), $\textrm{payload}=128$ bytes, and $\tau$ is usually in hundreds. Thus, $\beta+4\kappa/\tau\leq\lambda$ holds if $n\leq4X$.} 
we have that $\frac{(\beta+4\kappa/\tau)\cdot(n-1)}{\alpha}\leq1$. Hence, the value of $\frac{\Lambda_b^\Delta}{C^\Delta}$ approaches 1/2, especially in the large-scale environment when $n$ is large. 
It means that, when scaling up Leopard by adding $C^\Delta$ communication capacity at each replica, the throughput can be increased by $\frac{1}{2}C^\Delta$ at all scales of the protocol. 

We remark that, even if $f$ replicas conduct the selective attack when setting $\alpha=O(n\log n)$, the amortized communication complexity at each replica remains at $O(1)$. The communication cost, however, will increase by retrieving missing datablocks, and this will result in a less effective scaling up. 

\section{Experimental Evaluation}
\label{section:experiment}

We have implemented a prototype of Leopard in Golang, using threshold signatures based on threshold BLS \cite{BLS01,bls-code} and erasure codes based on Reed-Solomon erasure codes \cite{RScode60,rscode-code}. We evaluated the performance of Leopard compared to a state-of-the-art system, HotStuff\footnote{There are several other typical BFT systems like BFT-SMaRt \cite{BFTSmart} and SBFT \cite{SBFT19}. We presented a comparison of these protocols' amortized costs in $\S$\ref{section:compare}. The result in Table \ref{table:compare} shows that, the amortized cost of HotStuff is lower than that in BFT-SMaRt and SBFT (e.g., by the number of voting). Thus, we only evaluated the performance of Leopard with HotStuff.}, using the implementation by its authors \cite{HotStuff-code}. 
Our experiments were conducted on Amazon EC2 (c5.xlarge) instances in a single datacenter, each with 4 vCPUs and 9.8 Gbps TCP bandwidth. 
We ran each replica on a single EC2 instance. 
Every experimental result is averaged over 3 runs, each lasting until the measurement is stabilized. 

We ran various experiments, all with the largest possible fault tolerance (i.e., touching the 1/3 optimal resilience bound). We divide the experiments into four types: (i) throughput under a different number of replicas (up to 600) and varying batch sizes ($\S$\ref{exp:scalability}); (ii) the effectiveness of scaling up ($\S$\ref{exp:costeffectiveness}); (iii) breakdowns of consumptions across different components ($\S$\ref{exp:breakdown}); (iv) performance with Byzantine failures ($\S$\ref{exp:failure}). 
%
%
%
%
%
By default, the throughput is measured from the server's side
, and the latency is from the client's side. 

\subsection{Throughput on scaling}
\label{exp:scalability}
We examed the throughput of Leopard and HotStuff with an increasing number of replicas. The payload size is fixed at 128 bytes. We stress test the two systems with a saturated request rate and an increasing batch size to approach throughput's limit. 
Intuitively, increasing the batch size can better amortize the cost of confirmation and leads to higher throughput. Hence, we first measured throughput on varying batch sizes in both HotStuff and Leopard.  

\begin{figure}[!h]
	\centering{\includegraphics[width=6cm]{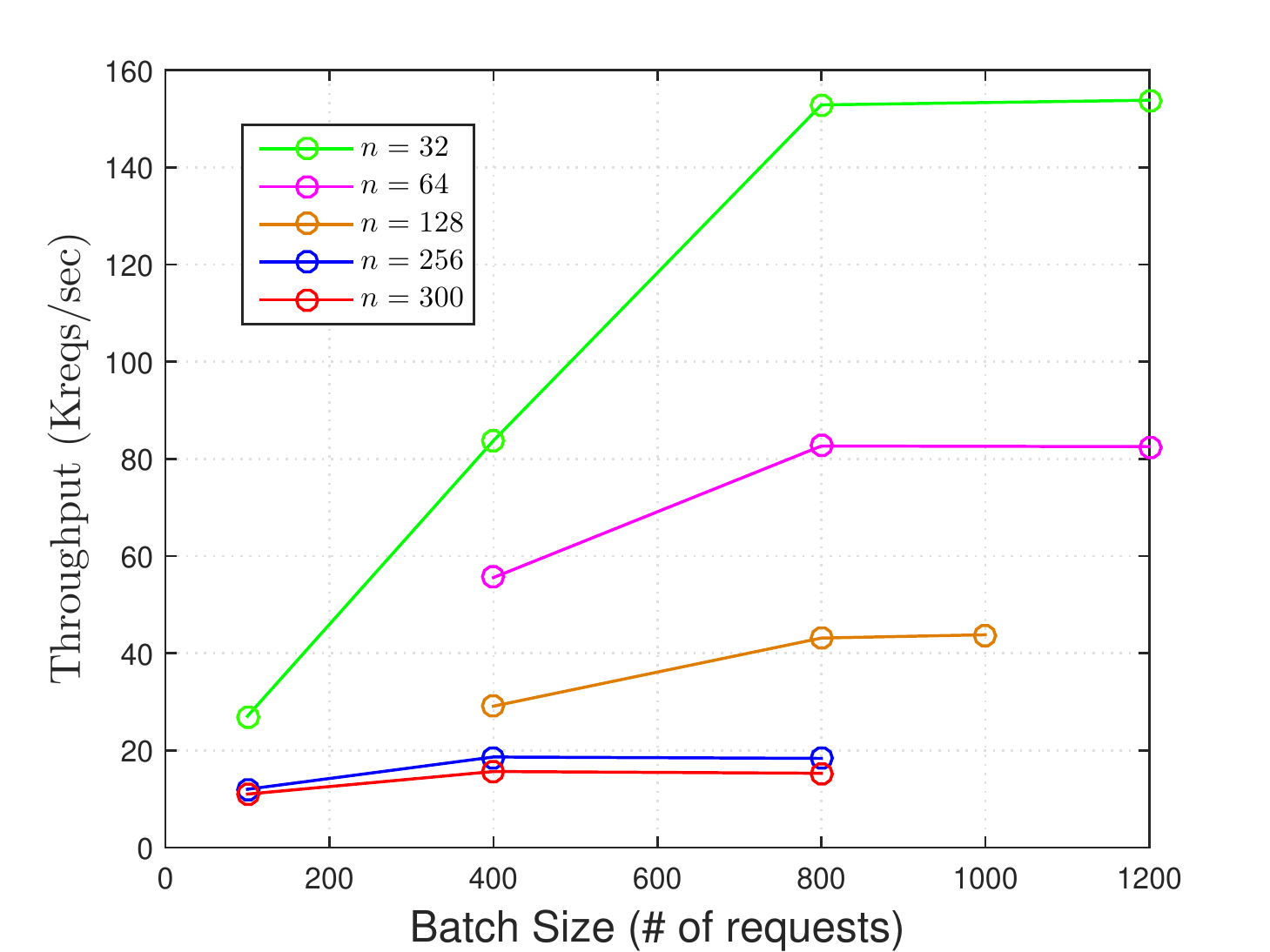}}
	\centering \caption{Throughput of HotStuff on varying batch sizes.}
	\label{fig:hotstuffpara}
\end{figure}

In HotStuff, there is only one batch for batching requests into a consensus proposal (i.e., a block). Fig. \ref{fig:hotstuffpara} depicts the throughput with different batch sizes for five replica numbers 32, 64, 128, 256, and 300. The result shows that throughput goes up as we increasing the batch size, but it stops to grow after a certain batch size value. 

\begin{figure}[ht]
	\centering		
	\includegraphics[width=0.35\textwidth,trim=0 0 0 0,clip]{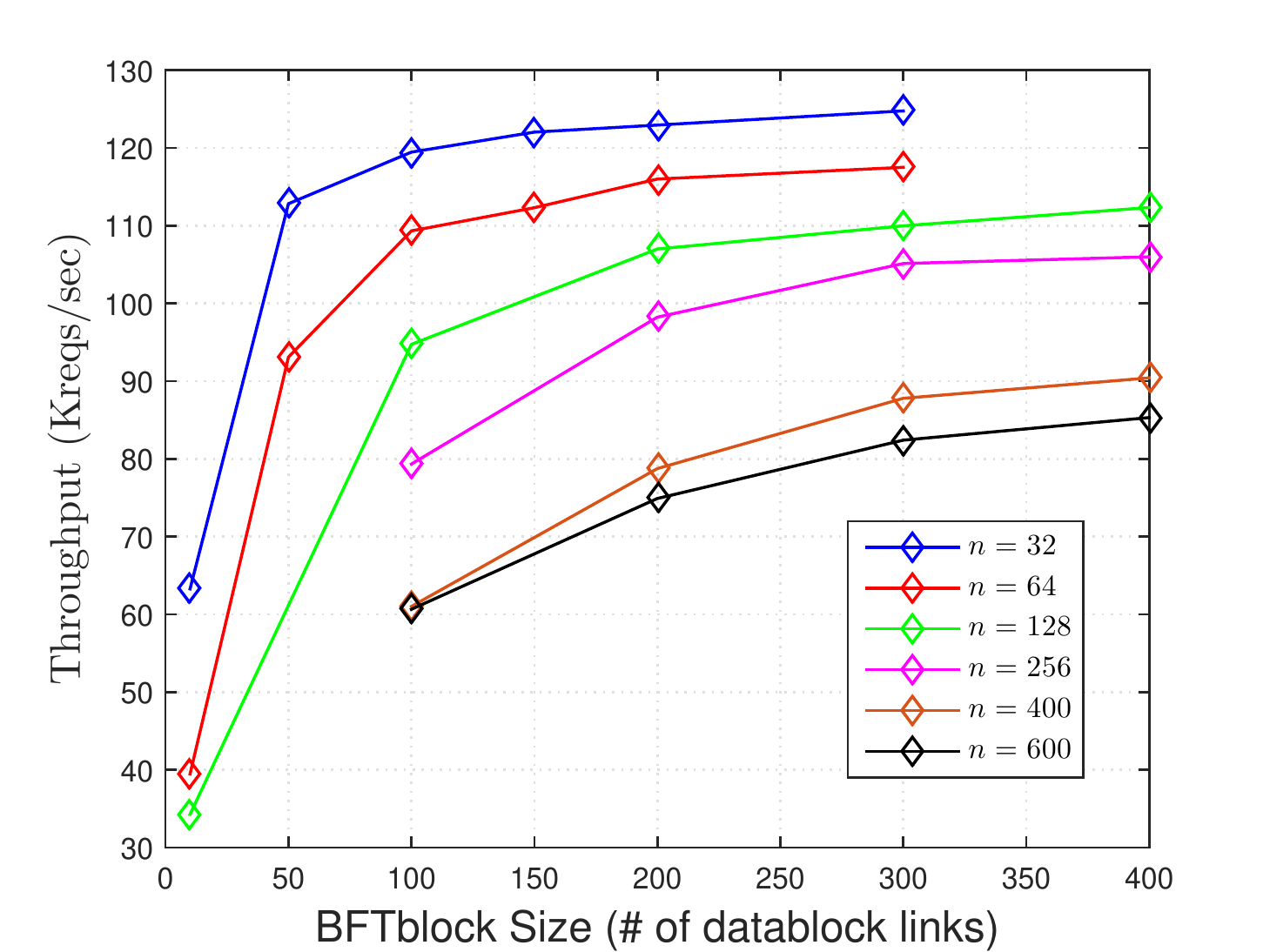}
	\caption{Throughput of Leopard on varying batch sizes of a BFTblock.}
	\label{fig:ourparabft}
\end{figure}

\begin{figure}
	\centering		
	\includegraphics[width=0.35\textwidth,trim=0 0 0 0,clip]{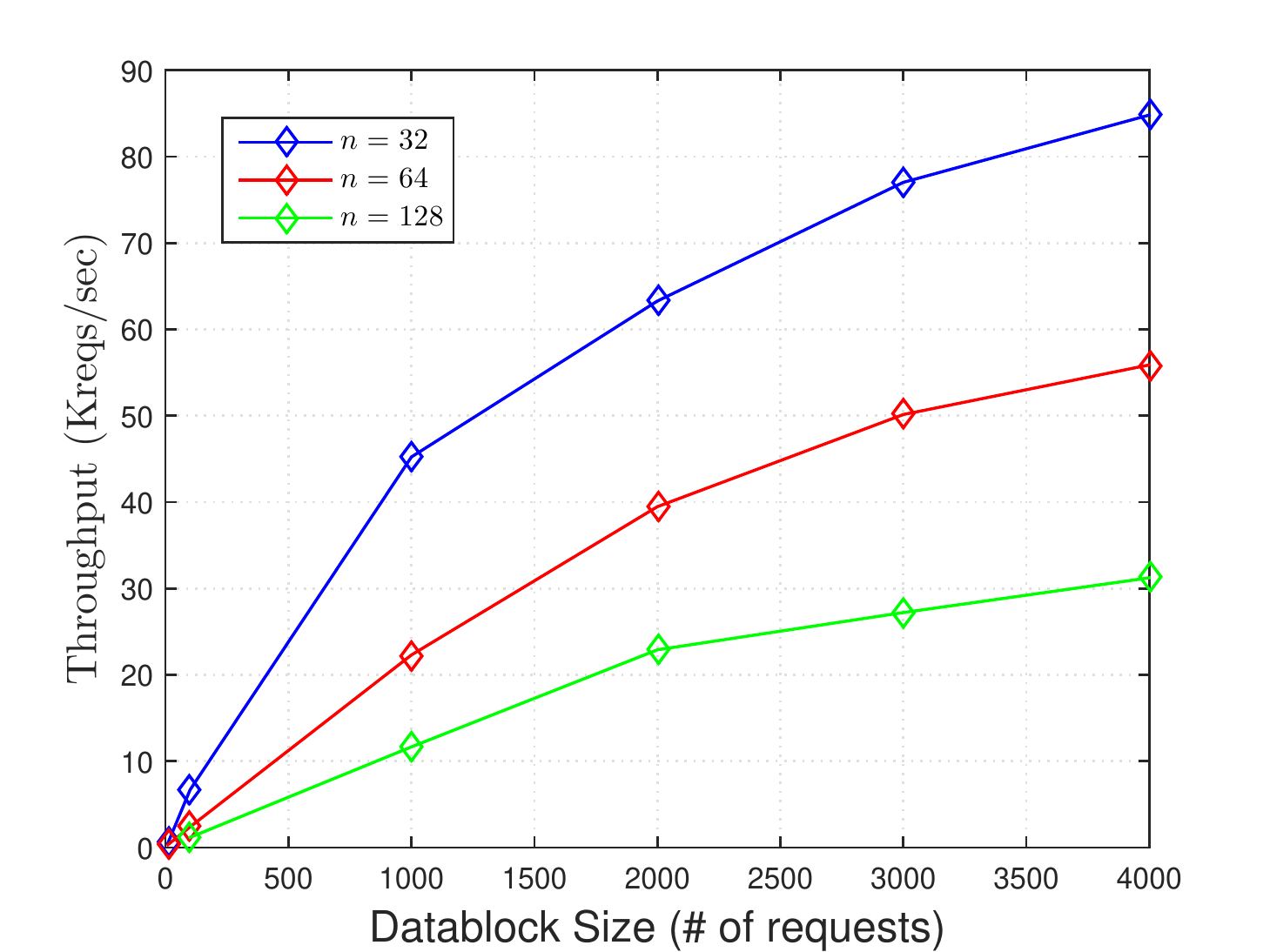}	
	\includegraphics[width=0.35\textwidth,trim=0 0 0 0,clip]{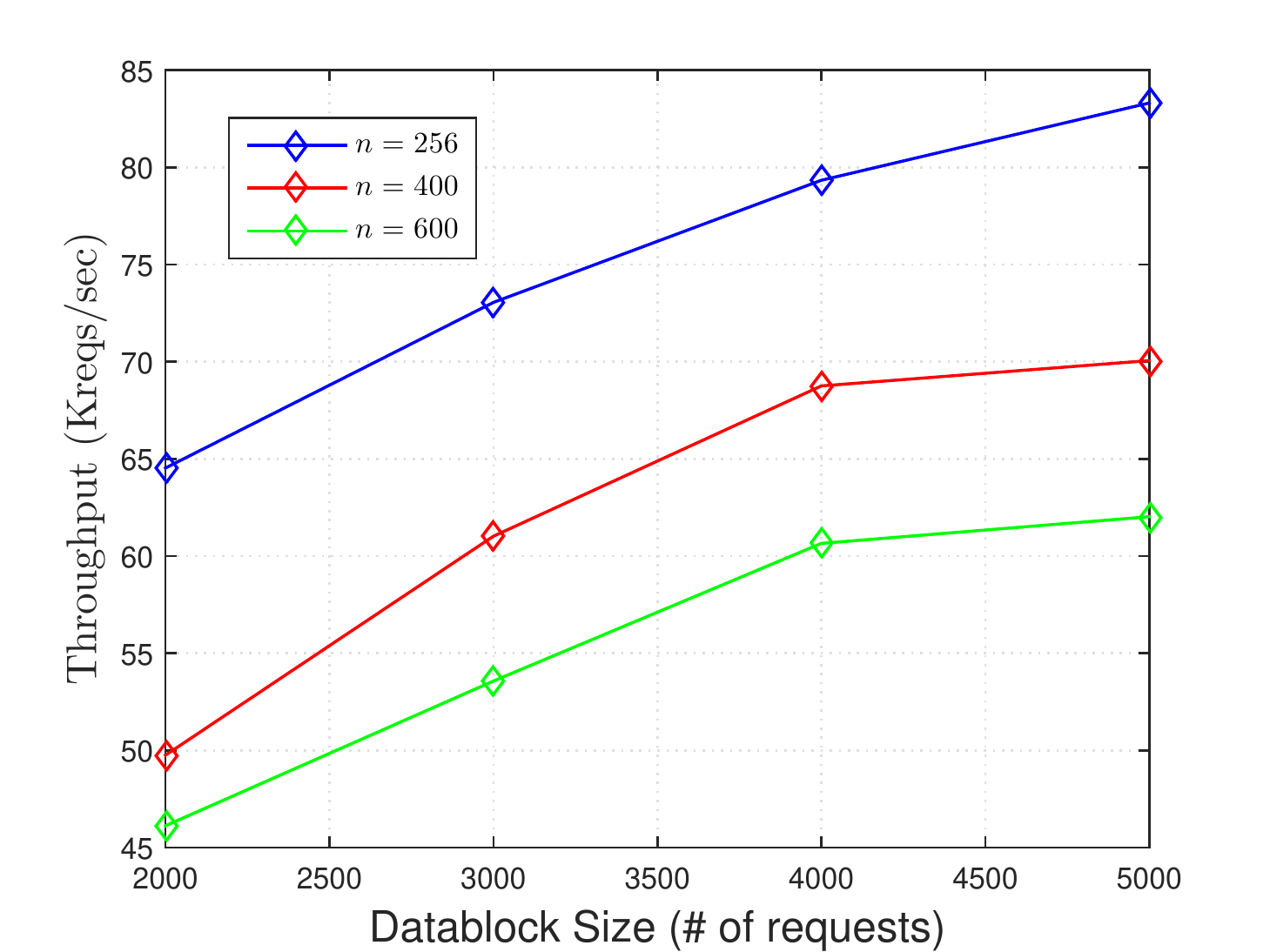}	
	\caption{Throughput of Leopard on varying datablock sizes with BFTblock size fixed at 10 (in the above) and 100 (in the below), respectively.}
	\label{fig:ourparatxb}
\end{figure}

In Leopard, there are two batch sizes, namely the BFTblock size and the datablock size. We varied one bath size while fixing the other one, to exam the effect on throughput. 
Fig. \ref{fig:ourparabft} depicts the throughput when the BFTblock size increases for six replica numbers (32 to 600). It shows that, for all tested $n$, the upward trend of throughput has gradually stabilized as the BFTblock size increases. Generally, the value of the BFTblock size, at which the throughput gets stabilized when $n$ is large, is greater than the value when $n$ is small. It indicates that it requires a larger batch size in each agreement instance to amortize the increased cost caused by a larger $n$. 
To exam the throughput when the datablock size increases, we chose two BFTblock sizes and see the result at different $n$. The result is shown in Fig. \ref{fig:ourparatxb}. Similar to the result in Fig. \ref{fig:ourparabft}, it shows that for both BFTblock sizes and all tested $n$, the upward trend of throughput has also gradually stabilized as the datablock size increases. (Although we doesn't show the latency in this section, we find out that the latency keeps going up as we increasing the batch sizes in both systems.)


We fine-tune the batch sizes to better accommodate both throughput and latency, and Table \ref{table:para} presents the batch sizes we used in this section.  
Fig. \ref{fig:scal} depicts throughput as $n$ grows 
in both Leopard and HotStuff. Since the implementation of HotStuff can hardly work when $n>300$, it only shows the result with $n$ up to 300. The throughput of Leopard drops only slightly and remains almost flat when $n$ varies up to 600, whereas HotStuff suffers from a significant throughput drop when its scale increases. When $n=300$, Leopard achieves a $5\times$ throughput compared to HotStuff, and the gap gets wider when $n$ further grows. 

\renewcommand{\arraystretch}{1.2}
\begin{table}[!h]
	\centering
	\caption{Implementation parameters of batch sizes}
	\vspace{-1mm}
	\footnotesize{
		\begin{tabular}{|c|c|c|c|}	
			\hline
			\multirow{2}{*}{$n$}& \multicolumn{2}{c|}{Leopard} & \multirow{2}{*}{HotStuff} \\
			\cline{2-3}
			& Datablock & BFTblock & \\
			\hline
			32&2000&100 & \multirow{5}{*}{800}\\
			\cline{1-3}
			64 & 2000 & 100 &\\
			\cline{1-3}
			128 & 3000 & 300&\\
			\cline{1-3}
			256 & 4000 & 300&\\
			\cline{1-3}
			300&-- &--&\\
			\hline
			400&4000&400&\multirow{2}{*}{--}\\
			\cline{1-3}
			600&4000&400&\\
			\hline
		\end{tabular}
	}
	\label{table:para}
\end{table}

\begin{figure}[!t]
	\vspace{-3mm}
	\centering
	\includegraphics[width=0.35\textwidth,trim=0 0 0 0,clip]{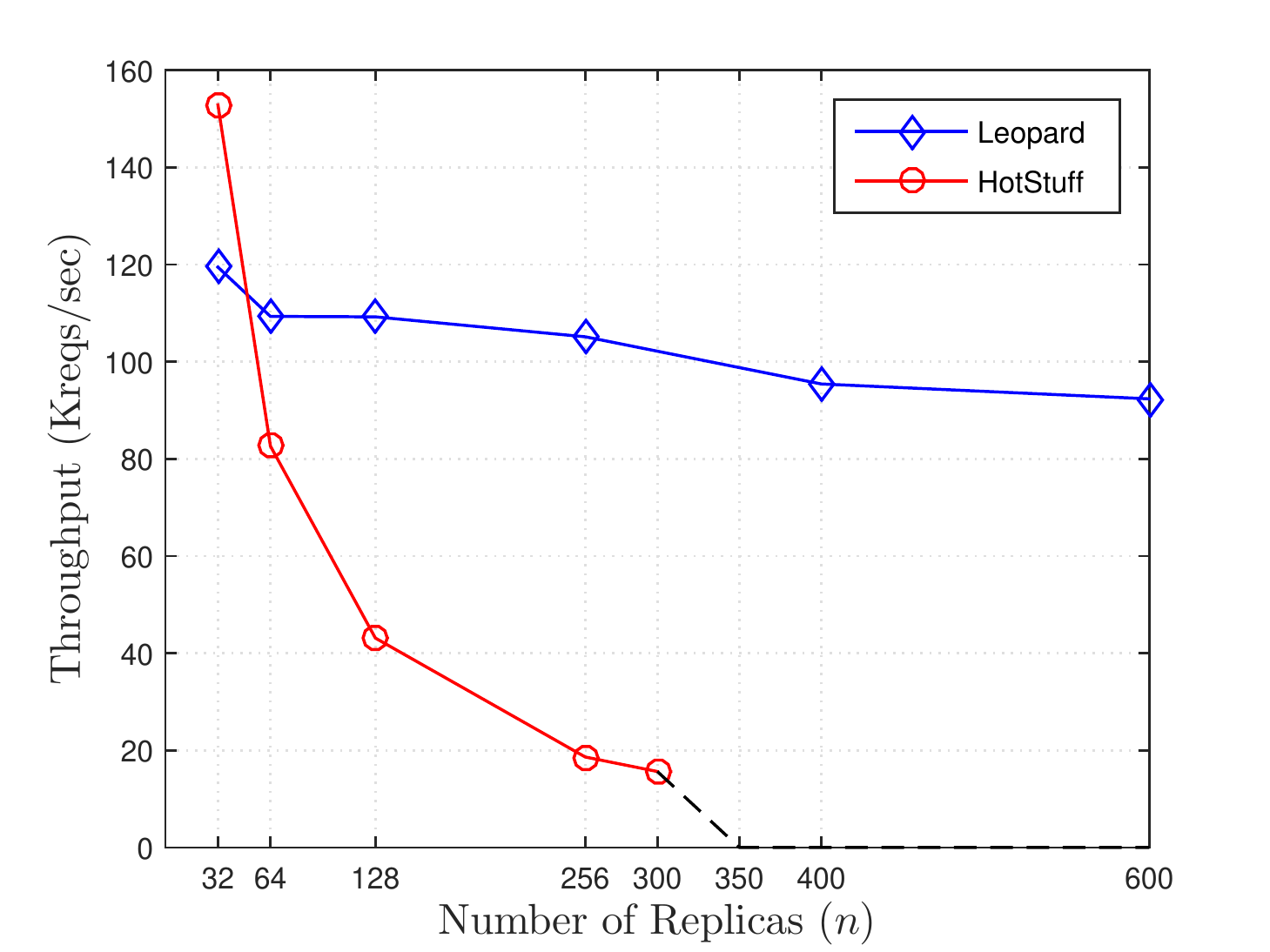}
	\vspace{-2mm}
	\caption{Throughput of Leopard and HotStuff at different scales.}
	\vspace{-4mm}
	\label{fig:scal}
\end{figure}

\begin{figure*}
	\vspace{-3mm}
	\centering
	\includegraphics[width=0.35\textwidth,trim=0 0 0 0,clip]{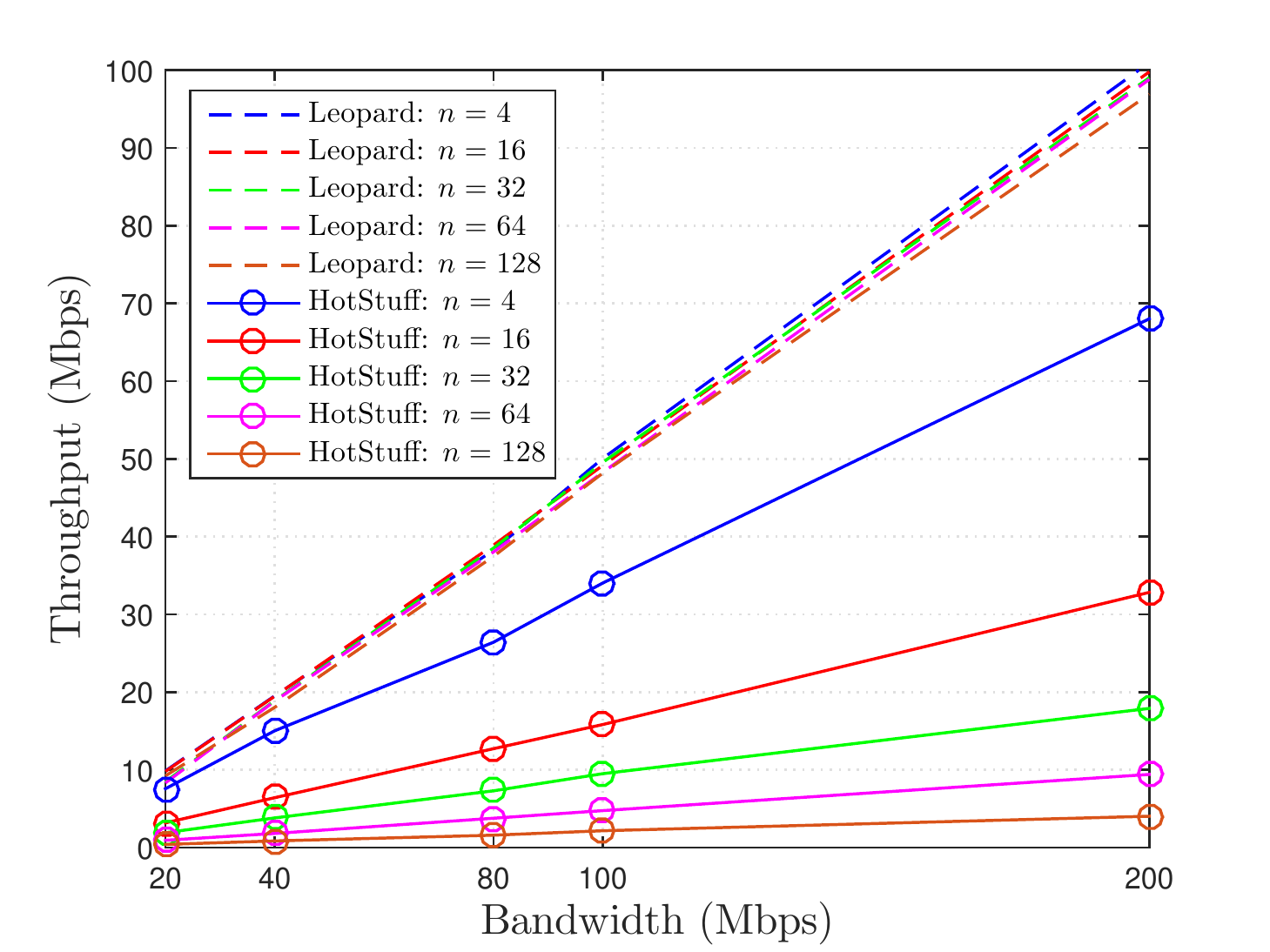}
	\includegraphics[width=0.35\textwidth,trim=0 0 0 0,clip]{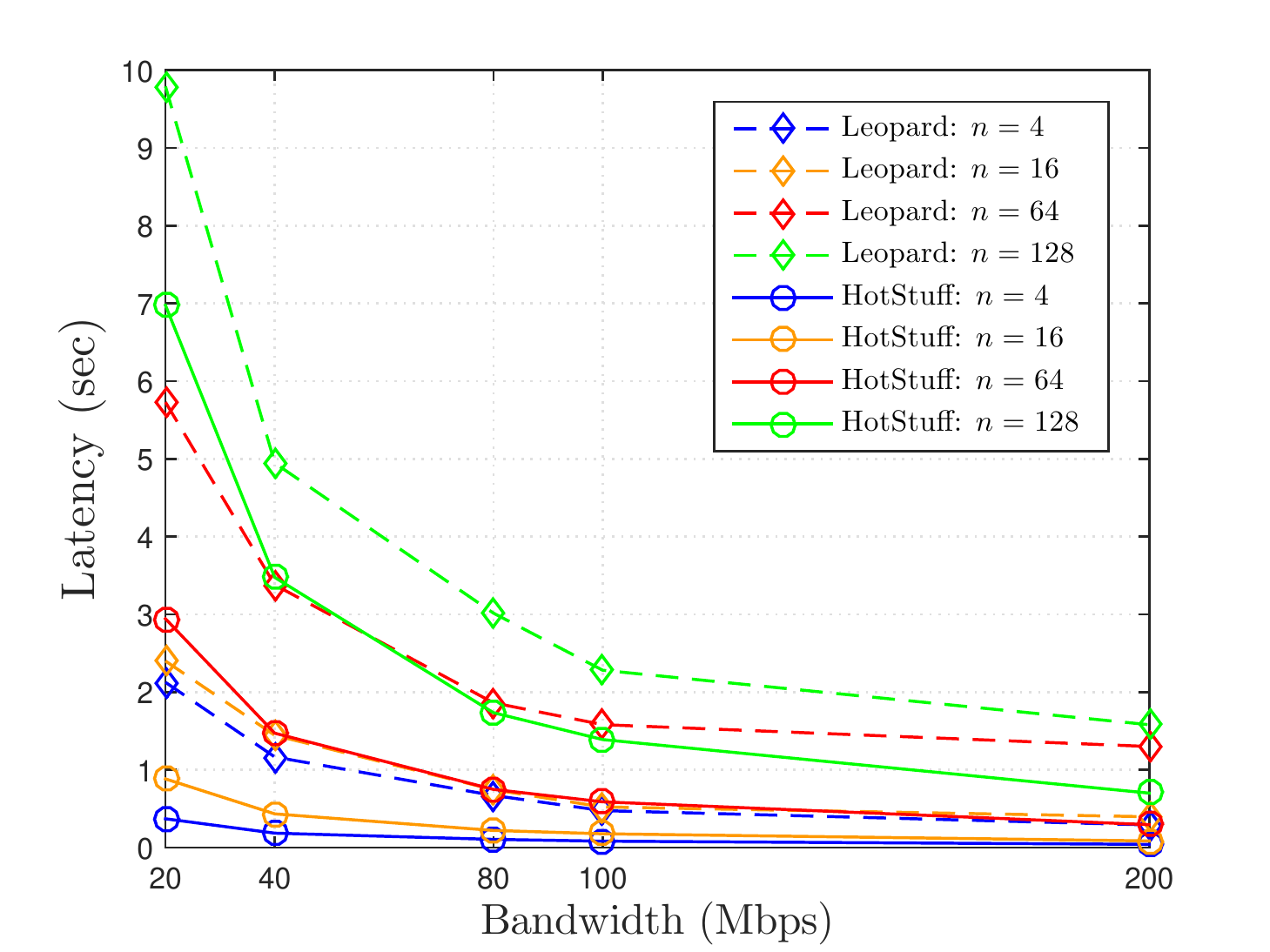}
	\vspace{-2mm}
	\caption{Throughput and latency with different available bandwidths at each replica in Leopard and HotStuff.}
	\vspace{-4mm}
	\label{fig:cost-effectiveness}
\end{figure*}

\subsection{Effectiveness of scaling up}
\label{exp:costeffectiveness}
We examed the effectiveness of the common method in distributed systems to improve throughput by scaling up. We evaluated how HotStuff performs with an increasing available bandwidth at each replica. This is to verify our analysis that the throughput's drop in leader-based BFT protocols that using the leader to disseminate requests cannot be effectively solved by just configuring with more resources. As a comparison, we evaluated the performance of scaling up in Leopard with different $n$.

We throttled the available bandwidth at each replica from 20 to 200 Mbps using NetEm \cite{NetEm}, and fixed the batch sizes in both two systems. 
Fig. \ref{fig:cost-effectiveness} depicts throughput and latency under different available bandwidths for both Leopard and HotStuff with a varying $n$. 
It shows that the throughput of both systems can be improved by configuring with more bandwidth resources, and it grows linearly with the increase of available bandwidth for all tested scales. In HotStuff, the increased throughput over the increased bandwidth, however, approaches 0 as $n$ gets larger, which indicates that throughput cannot be improved effectively by this method when the scale is large. 
In contrast, the improved throughput in Leopard remains at about 1/2 for all tested scales. 
The slight decrease of throughput in Leopard when $n$ varies from 4 to 128 is due to that we fixed the datablock size. But at all events, the effectiveness of scaling up on throughput in Leopard is significantly superior to HotStuff. 

A larger available bandwidth comes to lower latency for all tested scales in both systems. The latency of Leopard is higher than that of HotStuff in the same configuration. This is due to that 
replicas in Leopard wait for more requests' arrivals through receiving datablocks of size $O(n)$ before initiating a new agreement.
The latency gap, however, narrows down when adding a higher available bandwidth to reduce the time spent at request delivery. 

\renewcommand{\arraystretch}{1.2}
\begin{table}[!h]
	\vspace{-3mm}
	\centering
	\caption{Bandwidth utilization breakdown of Leopard}
	\vspace{-1mm}
	\footnotesize{
		\begin{tabular}{|c|c|l||c|}	
			\hline
			Role & \multicolumn{2}{c||}{Utilization} & \%Bandwidth\\
			\hline\hline
			\multirow{8}{*}{Leader}&\multirow{4}{*}{Send}&BFTblock& 3.17\%\\
			&& Proof& 0.24\%\\
			&&Miscellaneous & 0.15\%\\
			\cline{3-4}
			&&SUM & \textbf{3.56\%}\\
			\cline{2-4}
			&\multirow{4}{*}{Receive}&Datablock& 96.17\%\\
			&&Vote & 0.16\%\\
			&&Miscellaneous& 0.10\%\\
			\cline{3-4}
			&&SUM &\textbf{96.44\%}\\
			\hline			
			\multirow{9}{*}{Non-leader}&\multirow{4}{*}{Send}&Datablock& 49.93\%\\
			&& Vote & 0.01\%\\
			&&Miscellaneous& 0.05\%\\
			\cline{3-4}
			&&SUM & \textbf{49.99\%}\\
			\cline{2-4}
			&\multirow{5}{*}{Receive}&BFTblock& 0.05\%\\
			Replica&&Datablock & 48.34\%\\
			&&Proof& 0.01\%\\
			&&Reqs. from Clients& 1.61\%\\
			\cline{3-4}
			&&SUM & \textbf{50.01\%}\\
			\hline
		\end{tabular}
	}
\vspace{-1mm}
	\label{table:bandwidth-breakdown}
\end{table}

\subsection{Breakdowns}
\label{exp:breakdown}

\subsubsection{Bandwidth utilization}
To understand the bandwidth that is being utilized, so it can effectively improve protocol and implementation to achieve the best possible performance out of the network, 
we evaluated the bandwidth utilization breakdown at the leader and the non-leader replica in Leopard. Table \ref{table:bandwidth-breakdown} presents the evaluation result with $n=32$. It shows that, most of the bandwidth at the leader (over 96\%) is utilized on dealing with datablocks, whereas the cost of processing BFTblocks is only less than 4\%. 
For each non-leader replica, most of the bandwidth is utilized on processing datablocks, similar to the leader's side. The cost of sending datablocks (49.93\%) is a little bit higher than receiving (48.34\%). This is because a non-leader replica sends a datablock to $n-1$ other replicas, but it only receives $n-2$ replicas' datablocks. 
Moreover, the bandwidth used on processing votes is only less than 1\%. 
This implies that measuring the communication cost only on the vote phase is not enough for evaluating the efficiency of the protocol in the high throughput setting.


\subsubsection{Latency}
To understand the reason for the latency gap between Leopard and HotStuff shown in Fig. \ref{fig:cost-effectiveness}, we measured the time cost across different components in our Leopard’s implementation and provides a latency breakdown. The result is presented in Table \ref{table:latency-breakdown}. It shows that, our Leopard’s implementation costs over 50\% of latency on the datablock dissemination. Since this component only takes up 1/8 fraction of total rounds during an agreement (see Fig. \ref{fig:message-flow}), it indicates that engineering optimizations like more efficient data serialization and underlying communication mechanisms could further improve the performance, which we leave as future work. 

\renewcommand{\arraystretch}{1.2}
\begin{table}[h]
	\centering
	\caption{Latency breakdown of Leopard with $n=32$.}
	\vspace{-1mm}
	\small{
		\begin{tabular}{|c|l||c|}	
			\hline
			\multicolumn{2}{|c||}{Usage} & \%Latency\\
			\hline\hline
			\multirow{3}{*}{Datablock Preparation}&Datablock Generation& 12.98\%\\
			& Datablock Dissemination & 50.30\%\\
			\cline{2-3}
			&SUM & \textbf{63.28\%}\\
			\hline		
			\multicolumn{2}{|c||}{Agreement} & \textbf{35.97\%}\\
			\hline		
			\multicolumn{2}{|c||}{Response to the Client} & \textbf{0.75\%}\\
			\hline		
		\end{tabular}
	}
	\label{table:latency-breakdown}
\end{table}

\subsubsection{Bandwidth usage of the leader}
Since the leader's overhead is the major bottleneck in existing leader-based BFT protocols, we further evaluated the leader’s bandwidth utilization for both systems. The result is depicted in Fig. \ref{fig:leader-overhead}. It shows that the leader’s bandwidth in HotStuff rises rapidly from 1Gbps to 2.5Gbps when $n$ grows from 4 to 300 (at the same time, the throughput decreases dramatically). In Leopard, this has been greatly reduced to lower than 0.5Gbps for all tested scales. 

\begin{figure}[!h]
	\centering
	\includegraphics[width=0.35\textwidth,trim=0 0 0 0,clip]{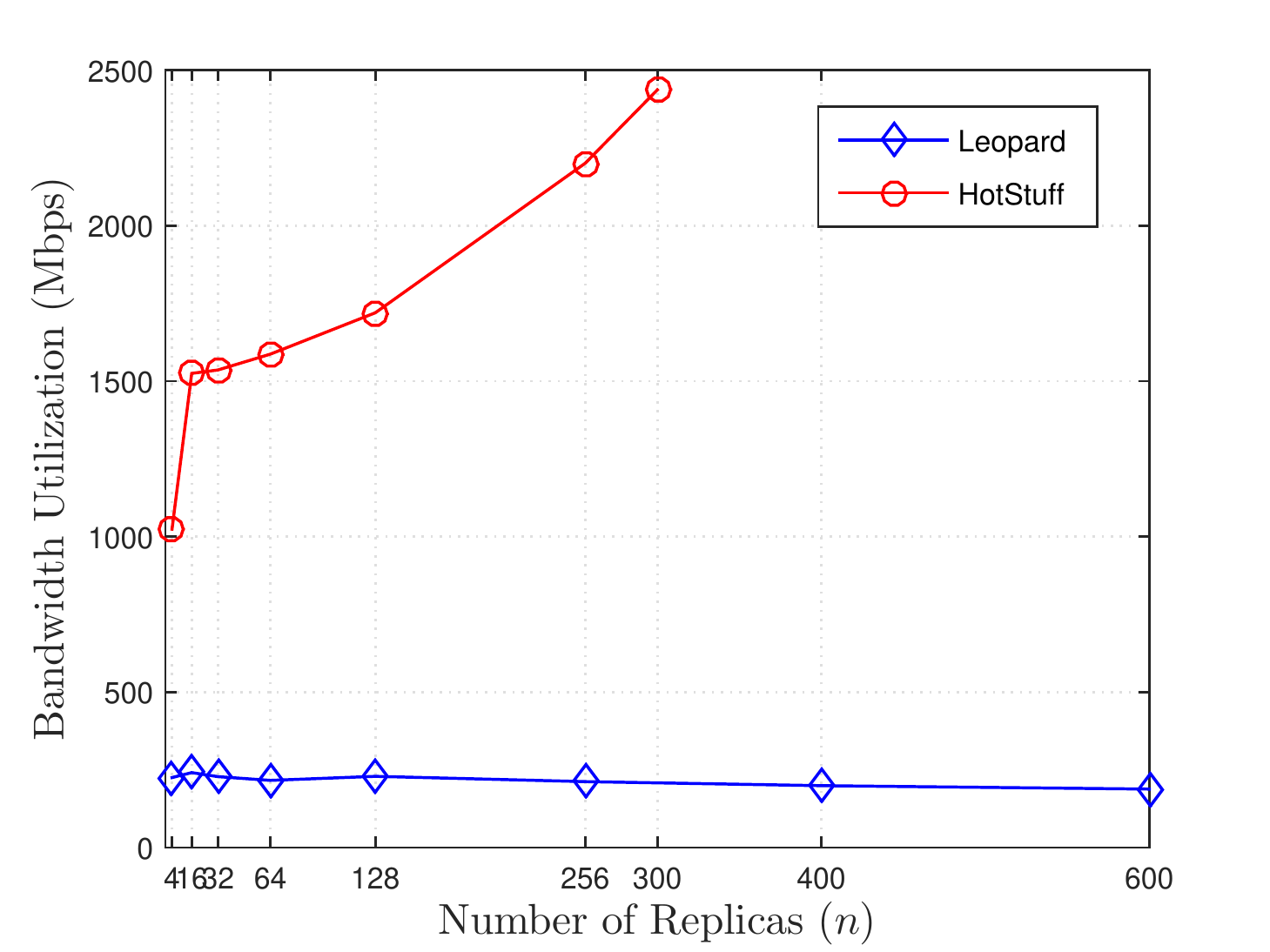}
	\vspace{-2mm}
	\caption{Bandwidth usage of the leader in Leopard and HotStuff.}
	\vspace{-2mm}
	\label{fig:leader-overhead}
\end{figure}

\subsection{Performance with Byzantine failures}
\label{exp:failure}

\subsubsection{Retrieving a missing datablock}
\begin{figure}
	\centering
	\includegraphics[width=0.35\textwidth,trim=0 0 0 0,clip]{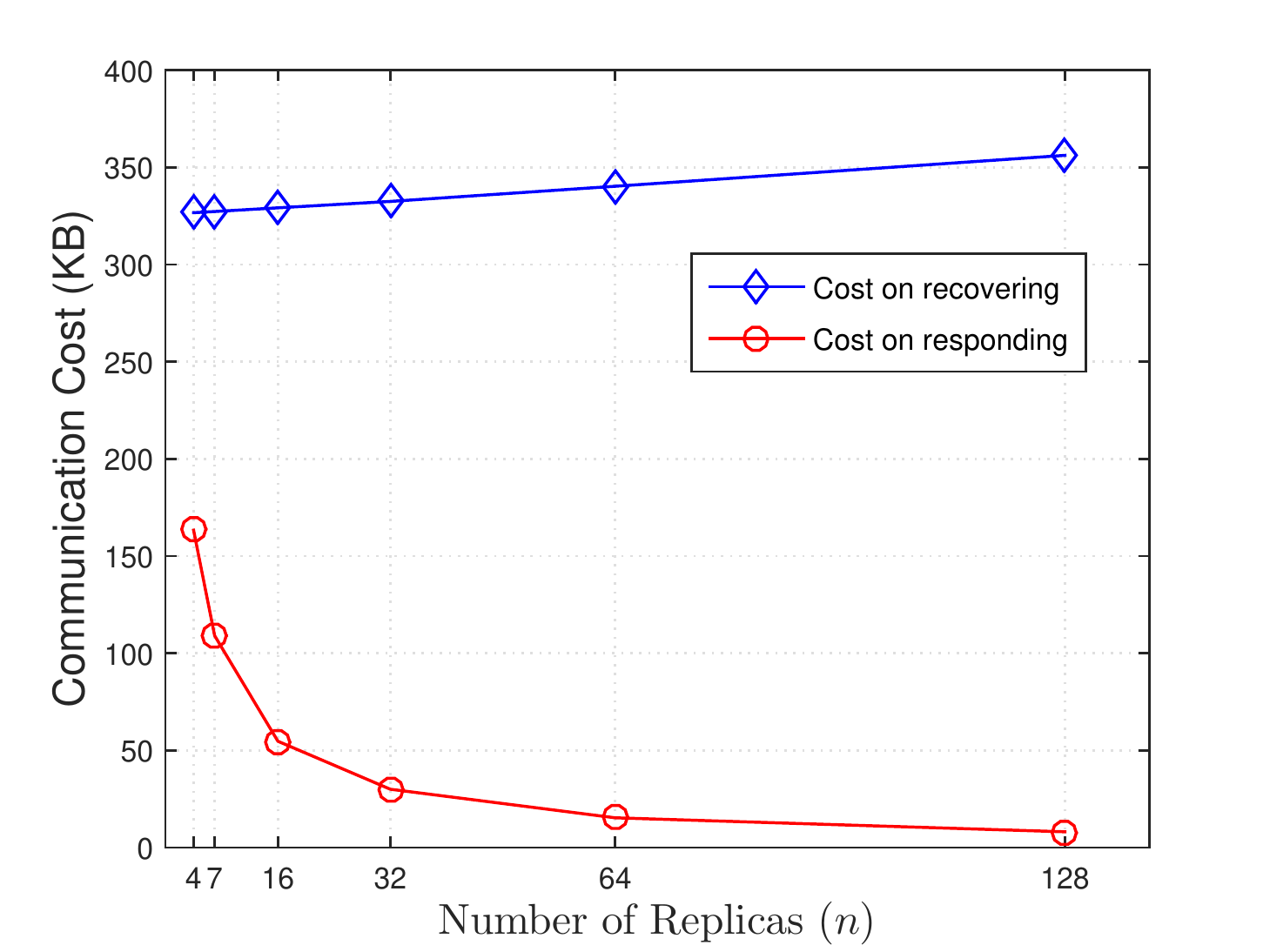}
	\caption{Communication cost with different numbers of replicas in Leopard.}
	\vspace{-2mm}
	\label{fig:retrievetest}
\end{figure}

\begin{table}
	\centering
	\caption{Time cost on retrieving a datablock at different $n$.}
	\small{
		\begin{tabular}{|c|cccccc|}	
			\hline
			$n$ & 4 & 7 & 16 & 32 & 64 & 128 \\
			\hline
			Time cost (ms) & 35 & 40 & 39 & 47 & 53 & 130 \\			
			\hline		
		\end{tabular}
	}
	\label{table:retreive-time-cost}
\end{table}

To test the performance in the face of failures, we evaluated the communication and time costs during the retrieval and the view-change. 
To retrieve a datablock containing 2000 requests (with 128-byte payload size), the communication cost on retrieving is depicted in Fig. \ref{fig:retrievetest}, and the time cost is presented in Table \ref{table:retreive-time-cost}. It shows that the cost on recovering a datablock is only slightly increased from 325KB to 356KB when $n$ grows from 4 to 128, and the cost of responding at a replica is dramatically reduced from 163KB to 8KB. This is due to the usage of erasure codes to amortize the cost among replicas. The time cost is increased, though, it only takes less than 130 ms when $n$ is 128.

\subsubsection{View-change}
\label{exp:viewchange}
It is generally conjectured that the heavy view-change may hinder the support for large-scale replicas. Meanwhile, a Byzantine leader can obstruct the confirmation of request and decreases the efficiency of the system by letting the protocol switch to the view-change mode. We thus measured the performance of the view-change in Leopard with up to 400 replicas to evaluate its time and communication costs. The result is measured after a view-change has been triggered since the expired time for triggering a view-change may vary in different applications. In our evaluation, we randomly stop the leader to trigger a view-change. 

Fig. \ref{fig:vc} depicts the results. It shows that when the protocol's scale increases, both the time and total communication cost increase. However, when the protocol scale goes up to hundreds (400), the time cost of a view-change can be still in seconds (less than 6s).
The total communication cost when $n=400$ is less than 100MB, and most of the communication cost is spent at the leader (for sending a new-view message of size $O(n)$ to all the other replicas). In addition, since the view-change is triggered randomly in this experiment, the number of outstanding BFTblocks varies and this leads to a varied evaluation result in different runs.


We note that the cost on a view-change is mainly for processing BFTblocks that are not included in the latest checkpoint. Although the view-change protocol of Leopard is mostly based on PBFT, the number of outstanding BFTblocks could be small since the number of requests linked by a BFTblock in Leopard is large. This helps to reduce the cost of a view-change. 
Also, since the view-change only happens occasionally, the above evaluation result indicates that Leopard can work well when the scale of the system is in hundreds. 

\begin{figure}
	\centering
	\includegraphics[width=0.35\textwidth,trim=0 0 0 0,clip]{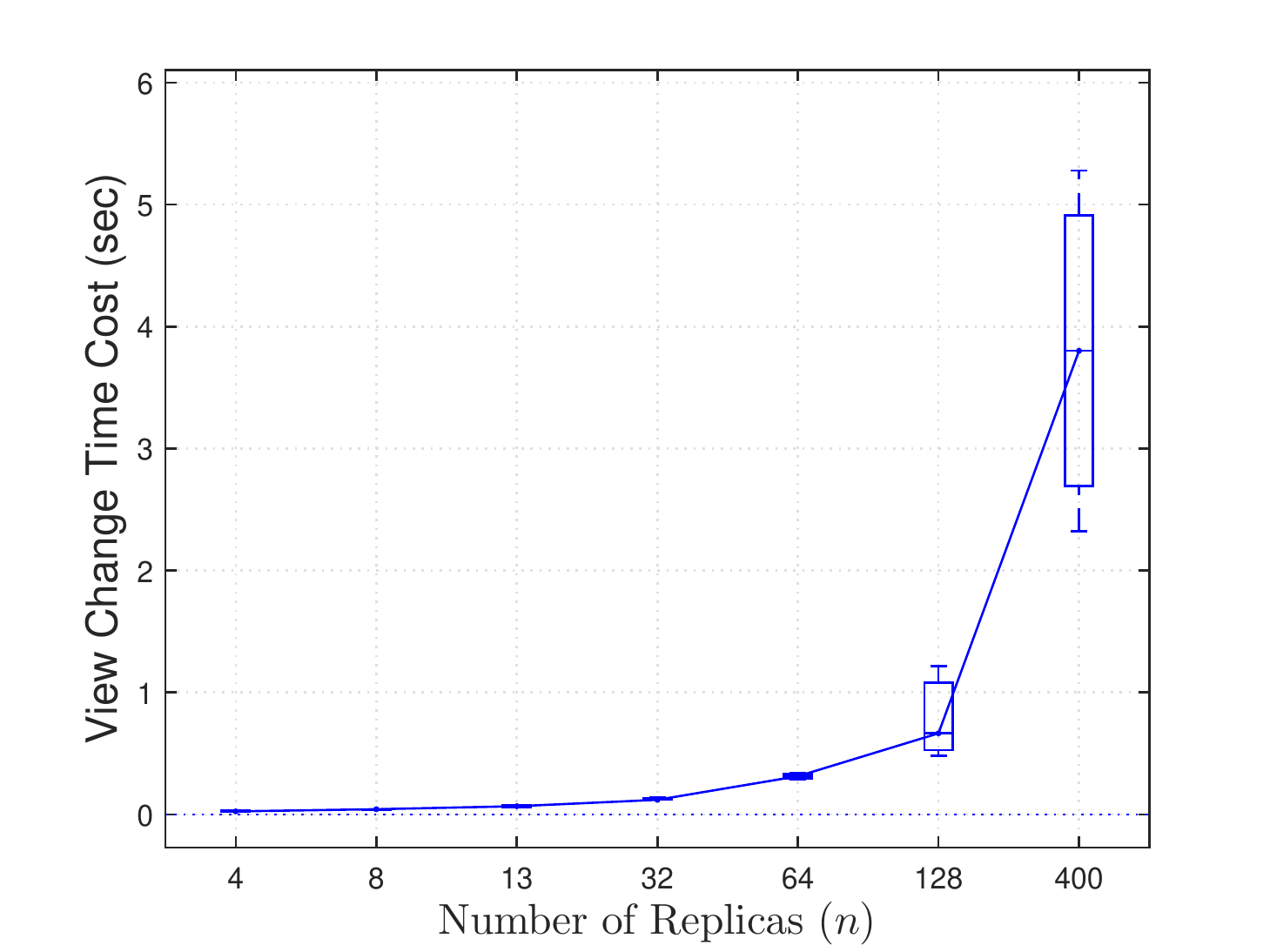}
	\includegraphics[width=0.35\textwidth,trim=0 0 0 0,clip]{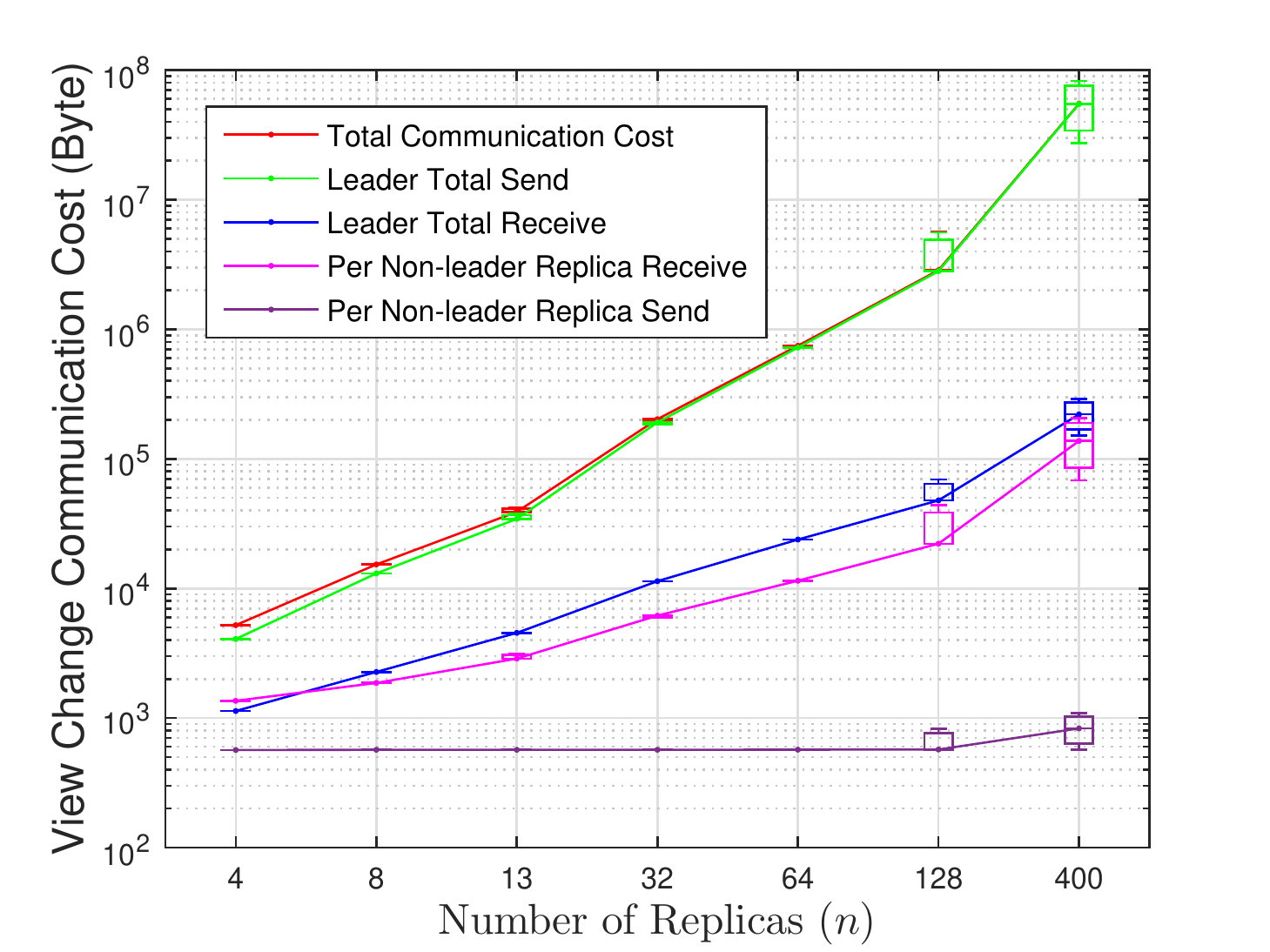}
	\vspace{-2mm}
	\caption{View-change time and communication costs of our Leopard's implementation.}
	\vspace{-4mm}
	\label{fig:vc}
\end{figure}

\bibliographystyle{IEEEtranS}
\bibliography{reference_all}

\appendices
\section{Garbage collection and view-change}
\label{appendix:detailprotocol}

\subsection{Garbage collection} 
To reduce the storage overhead of replicas, a.k.a., the garbage collection, it would be better if we could remove all executed requests from the buffer. 
We periodically (every $k/2$, where $k$ is the maximal number of parallel-executed BFTblocks in Algorithm \ref{alg:normal-case}) invoke a checkpoint protocol on the latest executed BFTblock, denoted as $B$. The generation of a valid checkpoint is a one-round voting process. After that, the requests linked by every BFTblock with the serial number lower than $B$'s can be removed from the buffer. Meanwhile, it advances the lower watermark $lw$ in Algorithm \ref{alg:normal-case} to the serial number of $B$. The checkpoint protocol is presented in Algorithm \ref{alg:checkpoint}. 
\begin{algorithm}[!h]
	\caption{Making a checkpoint}
	\small{
		\begin{algorithmic}[1] 
			\State (\textbf{for} replica $i,i\in[n]$)
			\State \quad let $sn$ be the serial number of the latest executed BFTblock and $st$ be the related execution state
			\State \quad \textbf{if} $(sn\mod k/2=0)$ \textbf{then} \cmd{produce a checkpoint}
			\State\qquad$cp\leftarrow\langle\texttt{checkpoint},sn,\mathsf{H}(st)\rangle$ 
			\State\qquad $\hat{\sigma}_{i}\leftarrow\mathsf{TSig}(tsk_i,cp)$
			\State \qquad send $(cp,\hat{\sigma}_{i})$ to $L_v$
			
			\vspace{1mm}
			\State (\textbf{for} leader $L_v$)
			\State \quad wait for $2f+1$ valid checkpoint messages
			\State\quad $\hat{\sigma}\gets\mathsf{TSR}(\{\hat{\sigma}_j\}_{j=1}^{2f+1})$ \cmd{create a proof for $cp$}			
			\State \quad multicast $(cp,\hat{\sigma})$ to all replicas
			
			\vspace{1mm}
			\State (\textbf{for} replica $i,i\in[n]$)
			\State\textbf{if} {$\mathsf{TVrf}(tpk,\hat{\sigma},cp)\rightarrow 1$} \textbf{then}
			\State \quad update the maintained latest checkpoint message $lc$ to $(cp,\hat{\sigma})$
			\State\quad$lw\leftarrow sn$ \cmd{advance the lower water mark $lw$}
		\end{algorithmic}
	}
	\label{alg:checkpoint}
\end{algorithm}

Like in Algorithm \ref{alg:normal-case}, we use the threshold signature to reduce the overall communication cost. For every replica, once an executed BFTblock has a serial number that is multiple of $k/2$, the replica generates a checkpoint message and sends it to the leader. The leader collects $2f+1$ checkpoint messages to form a proof for the checkpoint, and multicasts it to others. Each replica will update the latest confirmed checkpoint message and remove all executed requests linked by BFTblocks with a lower serial number than $sn$ from its buffer. The lower watermark will also be advanced to $sn$, indicating no BFTblock with a lower serial number will be received later. 

\subsection{View-change} 
The view-change mechanism aims to replace Byzantine leaders and reboot the normal-case mode at a new leader. 
The protocol contains three steps: Trigger a view-change, rotate the leader, and synchronize the state. View-change trigger decides when to invoke the view-change protocol; Leader rotation appoints a new leader; State synchronization synchronizes the state of the current view among replicas to avoid any safety-violation. 

We adopt a simple round-robin policy as in \cite{PaLa18,SBFT19,Ouroboros-BFT18} for the leader election, where the ($v$ mod $n$)-th replica be the eligible leader of view $v$. The view-change trigger and state synchronization steps in Leopard are similar to PBFT \cite{PBFT99}, which we show in the below:

\textit{- View-change trigger:} 
A view-change is triggered by a time-out request. Once the client finds out a request that has not received an acknowledgement for a long period of time, the client selects replicas following the protocol in Algorithm \ref{alg:req-assign}. It re-sends this request, together with a special tag indicating it is a time-out request, to the selected replicas, which are then disseminated to other replica through datablocks. Each replica sets a time on receiving a time-out request. The replica $i$ triggers a view change either (1) a timer expires and no acknowledgment of the request have been sent, and in this case, $i$ multicasts a timeout message $\langle\texttt{timeout},v\rangle$, together with $i$'s signature on it, to all replicas, or (2) $i$ receives a proof that the leader is faulty either via a publicly verifiable contradiction or when received $f+1$ timeout messages from other replicas. 
The replica then stops executing the protocols for agreeing BFTblocks (Algorithm \ref{alg:normal-case}) and making checkpoints (Algorithm \ref{alg:checkpoint}). 
As in PBFT, it requires that the timer for triggering a view-change should be set appropriately. This is to avoid switching to a new view too frequently considering a message delivery delay exists even if the network is synchronous.

\textit{- State synchronization:} 
During a view-change, each replica $i$ collects every notarized or confirmed BFTblock whose serial number is higher than $lw$. Replica $i$ sends a view-change message $m=\langle\texttt{view-change},v+1,lc,\mathcal{B}\rangle$, together with $i$'s signature on $m$, to the leader of the next view $v+1$. $\mathcal{B}$ is a set containing hashes of collected BFTblocks and their notarization proofs. 
The next view leader $L_{v+1}$ waits for $2f+1$ valid view-change messages. It then multicasts a new-view message $\hat{m}=\langle\texttt{new-view},v+1,\mathcal{V}\rangle$, together with $L_{v+1}$'s signature on $\hat{m}$, to all replicas, where $\mathcal{V}$ is the set of $2f+1$ view-change messages. 
When a replica receives a view-change message satisfying that it is signed by $L_{v+1}$ and contains $2f+1$ valid new-view messages, the replica advances to view $v+1$. Replicas restart the normal-case mode under view $v+1$ by first redoing the agreement for each BFTblock contained in $\mathcal{V}$ without re-executing them. The empty position between BFTblocks is filled by a dummy BFTblock with empty content.

\section{Safety and Liveness}
\label{appendix:security} 
The safety of the protocol is captured in Theorem \ref{theo:safety2}.
We first give and prove Lemma \ref{lem:uniqueness} and Lemma \ref{lem:safety} in the below. 

\begin{lemma}[Uniqueness]
	Assume the threshold signature scheme is secure. If two BFTblocks $B_1$ and $B_2$ both with the same serial number are notarized in the same view, then it must be the case that $B_1=B_2$.
	\vspace{-2mm}
	\label{lem:uniqueness}
\end{lemma}

\begin{proof}
	Since $B_1$ is notarized and by the definition of the protocol, there must be a set $R$ containing $2f+1$ distinct replicas who have signed $B_1$; otherwise, a reduction can be built to break the unforgeability of the threshold signature scheme. As there is at most $f$ Byzantine replicas, at least $f+1$ honest replicas are contained in $R$. Similarly, we have that at least $f+1$ honest replicas have signed $B_2$. Since there are $3f+1$ replicas in total and by the pigeon-hole principle, at least one honest replica has signed both $B_1$ and $B_2$. Due to that an honest replica will only sign one BFTblock for every serial number, it holds that $B_1=B_2$.
\end{proof}

\begin{lemma}
	If two honest replicas have two BFTblocks with the same serial number in their output logs, the two BFTblocks are the same.
	\label{lem:safety}
\end{lemma}

\begin{proof}
	We denote the two BFTblocks from the two honest replicas' logs as $B$ and $B'$, respectively. Due to Lemma \ref{lem:uniqueness}, if $B$ and $B'$ are added into the logs in the same view, then it holds that $B=B'$. We now consider the case where the two BFTblocks are added into the logs in two different views $v$ and $v'$, respectively. W.l.o.g., we assume $v<v'$.
	
	Following Algorithm \ref{alg:normal-case} (line 35), $B$ (resp. $B'$) is added into the log in view $v$ (resp. $v'$) only if it gets confirmed. 
	Hence, there must be $2f+1$ replicas knows that $B$ (resp. $B'$) is notarized. 
	During a view-change, a valid new-view message for advancing to the next view should contain $2f+1$ view-change messages from different replicas. Since there are at most $f$ Byzantine replicas, at least one honest replica $i$ satisfies (1) $B$ is notarized at $i$, and (2) the new-view message contains the view-change message sent by $i$. Hence, the protocol will redo BFTblock $B$ in the new view with the same serial number $sn$. Similarly, we deduce that the BFTblock $B$ added in the log in view $v$ is the same with the BFTBlock $B'$ added in the log in view $v'$. 
\end{proof}

\begin{theorem}[Restatement of Theorem 1]
    If two different requests exist on the same position of two output logs, then the two logs cannot be both from honest replicas.
	\label{theo:safety2}
\end{theorem}

\begin{proof}
	In Leopard, each BFTblock contains links to datablocks and each datablock contains a bunch of requests. If a BFTblock is added into a log of an honest replica, the linked request will be ordered in alphabetical order in the output log. Hence, Theorem \ref{theo:safety2} follows directly from Lemma \ref{lem:safety}.
\end{proof}
	
Due to the FLP impossibility \cite{FLP83}, the liveness of Leopard relies on the synchronous assumption, where synchrony is achieved after GST in the standard partially synchronous network model. The liveness of the protocol is captured in Theorem \ref{theo:liveness2}.

\begin{theorem}[Restatement of Theorem 2]
	After GST, the confirmation for a pending request will always be reached. 
	\label{theo:liveness2}
\end{theorem}

\begin{proof}
	In the case that the leader is honest, it will link every outstanding datablock in its generated BFTblocks, and initiate new agreement instances to confirm these BFTblocks. To ensure every request will eventually be confirmed, it implies that every request should be received by the leader and $f$ other honest replicas through some datablock, otherwise, the client will resend the request up to $f$ replicas, in which at least one honest replica will multicast the request in its datablock to all the replicas. If a datablock $m$ is linked by a BFTblock, according to our datablock retrieval mechanism (see Algorithm \ref{alg:retrieve2}), there must be at least $2f+1$ ready messages for $m$, and it contains at least $f+1$ honest replicas who have obtained $m$. 
	Hence, there must be enough (i.e., $f+1$) responses to help an honest replica who missed $m$ to recover $m$. Since there are at least $2f+1$ honest replicas, the confirmations of the honest leader's BFTblocks will be reached. 
	
	In the case that the current leader is faulty, who doesn't respond queryies and the confirmation of requests stops, the timer set by each replica will finally expires which will trigger a view-change to advance the protocol to the next view with a new leader. Since there are $f$ Byzantine replicas in total, the new leader will receive $2f+1$ new-view messages to generate a valid view-change message. Hence, a view-change can be completed. Also, after at most $f$ view-changes, the current leader must be an honest replica, and the liveness of the protocol is guaranteed. 
\end{proof}

\end{document}